\documentclass[12pt]{article}
\usepackage{amsfonts,array}
\usepackage{multicol,subcaption,makecell}
\usepackage{afterpage}
\usepackage{amsmath,bbm,dsfont,mathrsfs,mathtools,appendix, amsthm}
\usepackage{amssymb}
\usepackage{graphicx}
\usepackage{fullpage}%
\usepackage{enumitem}
\usepackage{framed}
\usepackage[table]{xcolor}
\usepackage{setspace}

\definecolor{shadecolor}{gray}{0.9}
\usepackage{authblk}

\usepackage[colorlinks=true,linkcolor=blue,citecolor=green,plainpages=false,pdfpagelabels]%
{hyperref}%
\hypersetup{
	colorlinks,
	linkcolor={blue!60!green},
	citecolor={green!50!yellow!75!black},
	urlcolor={blue!80!black},
	linktoc=all
}
\providecommand{\U}[1]{\protect\rule{.1in}{.1in}}

\newtheorem{theorem}{Theorem}
\newtheorem*{theorem*}{Theorem}

\newtheorem{corollary}[theorem]{Corollary}

\newtheorem{definition}[theorem]{Definition}

\newtheorem{lemma}[theorem]{Lemma}

\newtheorem{remark}[theorem]{Remark}

\newcommand{\reg}{\operatorname{reg}}
\newcommand{\cb}{\operatorname{cb}}
\newcommand{\cp}{\operatorname{cp}}
\newcommand{\im}{\operatorname{im}}
\newcommand{\scA}{\mathscr{A}}
\newcommand{\scB}{\mathscr{B}}

\newcommand{\Hil}{\mathcal{H}}

\newcommand{\M}[1]{\mathbb{M}_{#1}(\mathbb{C})}
\newcommand{\B}[1]{\mathcal{L}({#1})}
\newcommand{\iden}{\mathbbm{1}}

\newcommand{\id}{{\rm{id}}} %identity
\newcommand{\spec}{{\rm{spec}}}

\newcommand{\R}{\mathbbm{R}}
\newcommand{\C}{\mathbb{C}}

\newcommand{\N}{\mathbb{N}}

\newcommand{\cP}{\mathcal{P}}
\newcommand{\cQ}{\mathcal{Q}}

\newcommand{\cR}{\mathcal{R}}

\newcommand{\D}{\mathcal{D}}

\newcommand{\cW}{\mathcal{W}}

\newcommand{\cD}{\mathcal{D}}

\newcommand{\sD}{\widetilde{D}}

\def\>{{\rangle}}
\def\<{{\langle}}
\newcommand{\be}{\begin{equation}}
	\newcommand{\ee}{\end{equation}}
\newcommand{\bea}{\begin{eqnarray}}
	\newcommand{\eea}{\end{eqnarray}}

\newcommand{\floor}[1]{\left\lfloor {#1} \right\rfloor}

\newcommand{\ket}[1]{|#1\rangle} %ket
\newcommand{\bra}[1]{\langle#1|} %bra
\newcommand{\Tr}{\mathrm{Tr}}
\newcommand{\norm}[1]{\left\lVert #1 \right\rVert}
\newcommand{\abs}[1]{\left\lvert #1 \right\rvert}
\newcommand{\comment}[1]{}

\newcommand{\supp}{\operatorname{supp}}

\newcommand{\State}[1]{\mathcal{D}({#1})}

\numberwithin{equation}{section}
\numberwithin{theorem}{section}

\usepackage{color}
\definecolor{colorthree}{rgb}{0.01,0.51,0.93}

\usepackage{newunicodechar}
\newunicodechar{ρ}{\rho}
\newunicodechar{β}{\beta}
\newunicodechar{σ}{\sigma}
\newunicodechar{λ}{\lambda}
\newunicodechar{Λ}{\Lambda}
\newunicodechar{μ}{\mu}
\newunicodechar{ν}{\nu}
\newunicodechar{ψ}{\psi}
\newunicodechar{ϕ}{\varphi}
\newunicodechar{Φ}{\Phi}
\newunicodechar{φ}{\phi}
\newunicodechar{π}{\pi}
\newunicodechar{α}{\alpha}
\newunicodechar{ϵ}{\epsilon}
\newunicodechar{ε}{\varepsilon}
\newunicodechar{δ}{\delta}
\newunicodechar{ω}{\omega}
\newunicodechar{Δ}{\Delta}
\newunicodechar{Σ}{\Sigma}

\newcommand{\perr}{p_{\mathrm{err}}}
\newcommand{\n}{^{\otimes n}}

\title{Discriminating idempotent quantum channels}
\author[1]{Satvik Singh}
\author[2]{Bjarne Bergh}

\affil[1]{Department of Mathematics, Technical University of Munich, Garching, Germany}

\affil[1]{Munich Center for Quantum Science and Technology (MCQST), Munich, Germany}

\affil[2]{Department of Applied Mathematics and Theoretical Physics, University of Cambridge, United Kingdom}

\date{}

\begin{document}

\maketitle

\begin{abstract}
We study binary discrimination of idempotent quantum channels. When the two channels share a common full-rank invariant state, we show that a simple image inclusion condition completely determines the asymptotic behavior: when it holds, a broad family of channel divergences collapse to a closed-form, single-letter expression, regularization is unnecessary, and all error exponents (Stein/Chernoff/Hoeffding/strong-converse) are explicitly computable with no adaptive advantage. Crucially, this yields the strong converse property for this channel family. When the inclusion fails, asymmetric exponents become infinite, implying perfect asymptotic discrimination. We apply the results to GNS-symmetric channels, showing discrimination rates for large number of self-iterations converge exponentially fast to those of the corresponding idempotent peripheral projections. Without a common invariant state, we show that the same image inclusion still guarantees additivity of the sandwiched R\'{e}nyi cb-divergence for every $\alpha>1$, which suffices to establish the strong converse property and single-letter expressions for the Stein exponents.
\end{abstract}

\tableofcontents

\section{Introduction}

Quantum channel discrimination---the task of identifying an unknown quantum channel as one of two candidates given black-box access---is a central problem in quantum information theory \cite{chiribella_memory_2008,duan_perfect_2009,hayashi_discrimination_2009,harrow_adaptive_2010}, with applications spanning device verification, quantum communication, and benchmarking of quantum hardware. Despite extensive study, fundamental aspects of the problem remain poorly understood. In particular, the regularized channel divergence governing the optimal asymptotic asymmetric error exponent is generally intractable to compute \cite{wang_resource_2019, fang_chain_2020}, and it is unknown whether the strong-converse property holds for arbitrary channel pairs \cite{fawzi_defining_2021, Fang_2025towards}. Additionally, essentially nothing is known in general regarding the optimal asymptotic symmetric error exponents \cite{salek_usefulness_2022}.

In this work, we resolve many of these difficulties for the class of \emph{idempotent} quantum channels with full rank invariant states, i.e.\ channels $\cP : \B{\Hil} \to \B{\Hil}$ satisfying $\cP \circ \cP = \cP$ and such that $\cP(\iden_{\Hil})$ has full rank. This class includes conditional expectations onto unital $*-$subalgebras $\scA\subseteq \B{\Hil}$, replacer (constant) channels, projections onto decoherence-free subspaces and noiseless subsystems, and crucially, the asymptotic limits of quantum Markov semigroups satisfying suitable reversibility (or detailed balance) with respect to full-rank invariant states \cite{Gao2022sobolev}. \comment{Our main results show that if two such channels $\cP$ and $\cQ$ satisfy a natural image inclusion condition $\im (\cQ^*) \subseteq \im (\cP^*)$, and share a common invariant state $\tau = \cP(\tau)=\cQ(\tau)$, then (see Theorems~\ref{theorem:IDvQ}, \ref{theorem:PvsQ:common} and Corollaries~\ref{corollary:IdvQ}, \ref{corollary:PvQ:common-1}, \ref{corollary:PvQ:common-2}):
\begin{enumerate}[noitemsep]
    \item all quantum channel divergences of interest collapse to a single closed-form expression,
    \item regularization is unnecessary,
    \item all asymptotic error exponents are explicitly computable
    \item the strong converse property holds, and
    \item adaptive discrimination strategies offer no advantage over parallel ones.
\end{enumerate}
}

\subsection*{Summary of main results}

Our first main result, informally stated below, concerns idempotent channel pairs $\cP,\cQ$ that share a common full-rank invariant state and satisfy a natural image inclusion condition.

\newpage

\begin{shaded}
\begin{theorem*}
    Let $\cP, \cQ $ be two idempotent quantum channels with $\cP(\iden), \cQ(\iden)$ of full rank. If $\im (\cQ^*) \nsubseteq \im (\cP^*)$, then $D(\cP\|\cQ) = +\infty$, implying that the channels can be perfectly discriminated in the asymmetric setting (see Lemma~\ref{lemma:D(PQ)=inf}). Alternatively, if $\im (\cQ^*) \subseteq \im (\cP^*)$ and there exists a common full-rank invariant state $\tau = \cP(\tau)=\cQ(\tau)$, then all quantum divergences collapse (see Theorem~\ref{theorem:PvsQ:common}):
    \begin{equation*}
        D^{\cb}_{\min}(\cP \Vert \cQ) =  D^{\cb}(\cP\Vert \cQ) = D^{\cb}_{\max}(\cP\Vert \cQ),
    \end{equation*}
    where $D^{\cb}_{\min}, D^{\cb}$ and $D^{\cb}_{\max}$ are the stabilized min-, Umegaki, and max-channel divergences, respectively (see Section~\ref{sec:divergences} for precise definitions). Additionally, these divergences admit an explicit closed form expression (see \eqref{eq:D-closed-expression} in Theorem~\ref{theorem:PvsQ:common}).
    
    As a consequence, the following statements hold:
    \begin{enumerate}
        \item Additivity holds for all quantum R\'enyi divergences: $\mathbf{D}^{\cb}(\cP^{\otimes n} \|\cQ^{\otimes n}) = n\mathbf{D}^{\cb}(\cP\|\cQ)$, including for $D$, $D_{\max}$ and $D_{\min}$ (see the last part of Theorem~\ref{theorem:PvsQ:common}).
        \item All asymptotic error exponents (Stein/Chernoff/Hoeffding/strong-converse) for the task of discriminating these two channels (see Sections~\ref{sec:asymmetric-exponents}-\ref{sec:symmetric-exponents} for definitions) admit explicit single-letter expressions, and these exponents are equal for adaptive and parallel strategies (see Corollaries~\ref{corollary:PvQ:common-1}, \ref{corollary:PvQ:common-2} and Remark~\ref{remark:hoeffding} for precise statements).
        \item The strong converse property holds for the task of discriminating these two channels, i.e.~ the maximum achievable asymmetric type-II error exponent with any type-I error $ε < 1$ is equal to the asymmetric error exponent for vanishing type-I error $ε \to 0$ (see Corollary~\ref{corollary:PvQ:common-1} for the precise statement). 
    \end{enumerate}
\end{theorem*}
\end{shaded}

When the two channels do not share a common invariant state, such a drastic collapse of channel divergences does not occur. Nevertheless, we prove that the image inclusion condition still suffices to ensure that the stabilized sandwiched R\'{e}nyi channel divergence $\widetilde{D}^{\cb}_{\alpha}$ is additive for all $\alpha>1$ (see Theorem~\ref{theorem:PvsQ}), which suffices to establish the strong converse property and single-letter expressions for Stein exponents (see Corollary~\ref{corollary:PvsQ}).

\begin{shaded}
    \begin{theorem*}
For $j=1,2$, let $\cP_j,\cQ_j$ be idempotent channels with
$\cP_j(\iden)$ and $\cQ_j(\iden)$ of full rank, and
$\im(\cQ_j^*)\subseteq\im(\cP_j^*)$. Then, for every $\alpha>1$,
\begin{equation}
 \sD_\alpha^{\cb}(\cP_1\otimes\cP_2\|\cQ_1\otimes\cQ_2)
 =\sD_\alpha^{\cb}(\cP_1\|\cQ_1)
  +\sD_\alpha^{\cb}(\cP_2\|\cQ_2).
\end{equation}
\end{theorem*}
\end{shaded}

\subsection*{Outline}
Our paper is structured as follows.
\begin{itemize}
    \item Section \ref{sec:preliminaries} starts with an introduction of all the entropic quantities we use, as well as a summary of quantum hypothesis testing and previous results on quantum channel discrimination. 
    \item Section \ref{sec:main_results} contains our main theorems with their proofs.
    \item  Section \ref{sec:application_gns} presents an application of our results to the class of GNS symmetric channels.
\end{itemize}

\section{Preliminaries}\label{sec:preliminaries}
We denote quantum systems and their associated complex Hilbert spaces\footnote{All Hilbert spaces are assumed to be finite-dimensional throughout this work.} by capital letters $A,B,C$, with corresponding dimensions $d_A, d_B$ and $d_C$, respectively. The space of linear operators acting on a Hilbert space $\Hil$ is denoted by $\B{\Hil}$ and the convex set of quantum states (i.e.~positive semi-definite operators in $\B{\Hil}$ with unit trace) is denoted by $\cD (\Hil)$. For a unit vector $\ket{\psi}\in \Hil$, we denote the pure state $\ket{\psi}\bra{\psi}\in \cD (\Hil)$ by $\psi$. The identity operator in $\B{\Hil}$ is denoted by $\iden_{\Hil}$. The trace-norm, Schatten $\alpha-$norm for $\alpha\in (1,\infty)$, and the operator norm on $\B{\Hil}$ are denoted by $\norm{\cdot}_1, \norm{\cdot}_{\alpha},$ and $\norm{\cdot}_{\infty}$, respectively \cite{bhatia_matrix_1997}.

A quantum channel $\Phi:\B{A}\to \B{B}$ is a linear, completely positive, and trace-preserving map. The adjoint $\Phi^*: \B{B}\to \B{A}$ of a quantum channel $\Phi : \B{A}\to \B{B}$ is a linear, unital, and completely positive map defined through the following relation: 
\begin{equation}
    \forall X \in \B{A}, \, \forall Y \in \B{B} : \quad \Tr(Y \Phi(X))= \Tr (\Phi^*(Y) X).
\end{equation}

\subsection{Entropies and Divergences} \label{sec:divergences}
\begin{definition}\label{def:divergence}
    Let $\rho\in \State{\Hil}$ be a state and $\sigma\in \B{\Hil}$ be a positive semi-definite operator.
    \begin{itemize}
        \item The \emph{quantum (Umegaki) relative entropy} between $\rho$ and $\sigma$ is defined as \cite{umegaki_conditional_1962}
\begin{align*}
    D(\rho \Vert \sigma) := \begin{cases}
        \Tr \rho (\log \rho -\log \sigma) \quad &\text{if } \supp \rho \subseteq \supp \sigma \\
        + \infty  &\text{otherwise}
    \end{cases}
\end{align*}
        \item The $\alpha$-\emph{Petz R\'enyi relative entropy} between $\rho$ and $\sigma$ with $\alpha\in (0,1)\cup (1,\infty)$ is defined as \cite{petz_quasi-entropies_1986, Petz1985relative}
\begin{align*}
    D_{\alpha}(\rho \Vert \sigma) := \begin{cases}
        \frac{1}{\alpha-1} \log \Tr \left[ \rho^{\alpha} \sigma^{1-\alpha}  \right] \quad &\begin{cases}
            \text{if } \alpha\in (0,1) \text{ and } \rho\sigma\neq 0 \\
            \text{or } \alpha\in (1,\infty) \text{ and } \supp \rho \subseteq \supp \sigma
        \end{cases}  \\
        + \infty &\text{otherwise}
    \end{cases}
\end{align*}
        \item The $\alpha$-\emph{sandwiched R\'enyi relative entropy} between $\rho$ and $\sigma$ with $\alpha\in (0,1)\cup (1,\infty)$ is defined as \cite{muller-lennert_quantum_2013, wilde_strong_2014}
\begin{align*}
    \widetilde{D}_{\alpha}(\rho \Vert \sigma) := \begin{cases}
        \frac{1}{\alpha-1} \log \Tr \left[ \left(\sigma^{\frac{1-\alpha}{2\alpha}} \rho \sigma^{\frac{1-\alpha}{2\alpha}} \right)^{\alpha}  \right] \quad &\begin{cases}
            \text{if } \alpha\in (0,1) \text{ and } \rho\sigma \neq 0 \\
            \text{or } \alpha\in (1,\infty) \text{ and } \supp \rho \subseteq \supp \sigma
        \end{cases} \\
        + \infty &\text{otherwise}
    \end{cases}
\end{align*}
    \item The \emph{max-relative entropy} between $\rho$ and $\sigma$ is defined as \cite{datta_max-_2009}
\begin{align*}
    D_{\max}(\rho \Vert \sigma) := \begin{cases}
        \log \norm{\sigma^{-1/2} \rho \sigma^{-1/2} }_{\infty} \quad &\text{if } \supp \rho \subseteq \supp \sigma \\
        + \infty &\text{otherwise}
    \end{cases}
\end{align*}
    \item The \emph{min-relative entropy} between $\rho$ and $\sigma$ is defined as \cite{datta_min-_2009}
\begin{align*}
    D_{\min}(\rho \Vert \sigma) := \begin{cases}
        -\log  \Tr (\Pi_{\rho} \sigma) \quad &\text{if } \rho\sigma\neq0 \\
        + \infty &\text{otherwise},
    \end{cases}
\end{align*}
    where $\Pi_{\rho}$ denotes the orthogonal projection onto the support of $\rho$.
    \end{itemize}
\end{definition}

\begin{remark}\label{remark:Dmax-operator-order}
    The max-relative entropy can be alternatively expressed as \cite{datta_max-_2009}
    \begin{equation}
        D_{\max}(\rho \Vert \sigma) = \log \inf \{\lambda : \rho \leq \lambda \sigma \},
    \end{equation}
    where, for Hermitian operators $X,Y\in \B{\Hil}$, $X\leq Y$ means $Y-X$ is positive semi-definite, and the infimum over an empty set is assumed to be $+\infty$.
\end{remark} 

These divergences are ordered as follows: for $0<\alpha\leq \beta <1$ and $1< \alpha'\leq \beta' < \infty$, 
\begin{equation}
    D_{\min}(\cdot \| \cdot) \leq D_{\alpha}(\cdot \| \cdot) \leq D_{\beta}(\cdot \| \cdot) \leq D (\cdot \| \cdot)\leq \sD_{\alpha'}(\cdot \| \cdot) \leq \sD_{\beta'}(\cdot \| \cdot) \leq D_{\max}(\cdot \| \cdot).
\end{equation}
Moreover, $\lim_{α \to 0} D_α(\cdot\|\cdot) = D_{\min}(\cdot\|\cdot), \lim_{\alpha\to 1} D_{\alpha}(\cdot \| \cdot)=D(\cdot \| \cdot)= \lim_{\alpha\to 1} \sD_{\alpha}(\cdot \| \cdot)$, and \\ $\lim_{α \to \infty} \widetilde{D}_α(\cdot \| \cdot) = D_{\max}(\cdot \| \cdot)$ (see e.g.\ \cite{khatriPrinciplesQuantumCommunication2024}).

For any divergence $\mathbf{D}$ from Definition~\ref{def:divergence}, and two channels $\Phi, \Psi: \B{A}\to \B{B}$, we also define the associated channel divergence as
\begin{equation}\label{eq:Ddef}
    \mathbf{D}(\Phi \Vert \Psi) \coloneqq \sup_{\nu \in \D(A)} \mathbf{D}(\Phi(\nu)\|\Psi(\nu))\,.
\end{equation}
Using this, we can define the completely bounded (or stabilized) and regularized channel divergences as
\begin{align}
    \mathbf{D}^{\cb}(\Phi\|\Psi) &\coloneqq \mathbf{D}(\id_{A'} \otimes \Phi\|\id_{A'} \otimes \Psi) \label{eq:Dcb-def} \\
    \mathbf{D}^{\cb,\reg}(\Phi\|\Psi) &\coloneqq  \lim_{n \to \infty}  {\frac 1 n} \mathbf{D}^{\cb}(\Phi^{\otimes n} \|\Psi^{\otimes n}) = \lim_{n \to \infty}  {\frac 1 n} \mathbf{D}((\id_{{A'}^n} \otimes \Phi^{\otimes n}) \|(\id_{{A'}^n} \otimes \Psi^{\otimes n})), \label{eq:Dcbreg-def}
\end{align}
where $A'\cong A$ and $\id_{A'}$ denotes the identity channel on $\B{A'}$.

\subsection{Quantum Hypothesis Testing}
Hypothesis testing deals with the question of asserting the truth of one of multiple possible hypotheses given some object or data. In the simplest case, the given object is a quantum state and the task is to identify which of two (fully specified) options it is. Generically, if there are two hypotheses, there are two ways of making an error (mistaking the first for the second, or mistaking the second for the first). Given a decision strategy, we will call the corresponding probabilities of making such an error the type-I and type-II error probabilities. In the simple case where the task is to identify a given state as either $ρ\in \cD(\Hil)$ or $σ\in \cD (\Hil)$, these two probabilities are given by
\begin{align}
α &\coloneqq \mathbb{P}[\text{we claim the state is $σ$ } | \text{ state is actually $ρ$}]  \qquad \qquad (\text{type-I error),}\\
β &\coloneqq \mathbb{P}[\text{we claim the state is $ρ$ } | \text{ state is actually $σ$}]  \qquad \qquad \text{(type-II error).}
\end{align}

The most general way to arrive at such a decision in the case of distinguishing $ρ$ and $σ$ is by performing a binary (i.e., a two-outcome) POVM measurement, which is fully specified by one of its elements $0 \leq M \leq \iden_{\Hil}$, and we use the convention that an outcome corresponding to the measurement $M$ (resp.~$\iden_{\Hil} - M$) leads to the inference that the state is $\rho$ (resp.~$\sigma$). In that case the corresponding error probabilities are given by
\begin{align}
α &= \Tr((\iden_{\Hil} - M) ρ) = 1 - \Tr(Mρ)\\
β &= \Tr(M σ) \,.
\end{align}

We can then optimize over $M$ to find the optimal measurement in a certain sense. In the so-called asymmetric setting of hypothesis testing the aim is to minimize the type-II error probability given the constraint that the type-I error is below a chosen threshold $ε \in [0,1]$:
\begin{equation}
    \min_{M:\, 0 \leq M \leq \iden_{\mathcal{H}}}\{\Tr(M σ) \, | \, 1 - \Tr(M ρ) \leq ε\} = \min_{\substack{M:\, 0 \leq M \leq \iden_{\mathcal{H}} \\ \Tr(M \rho) \geq 1- \varepsilon}} \Tr(M σ).
\end{equation}

The negative logarithm of this is called the \emph{hypothesis testing relative entropy} \cite{wang_one-shot_2012}:
\begin{equation}\label{eq:hypothesis_testing_relative_entropy}
    D_H^{\varepsilon}(\rho\|\sigma) \coloneqq - \log \left[\min_{\substack{M:\, 0 \leq M \leq \iden_{\mathcal{H}} \\ \Tr(M \rho) \geq 1- \varepsilon}} \Tr(M σ)\right] = \max_{\substack{M:\, 0 \leq M \leq \iden_{\mathcal{H}} \\ \Tr(M \rho) \geq 1- \varepsilon}} - \log(\Tr(M σ))\,.
\end{equation}
It has been given this name, since it shares some properties with the quantum relative entropy, in particular, it satisfies the data-processing inequality \cite{wang_one-shot_2012}. 
Famously, the hypothesis testing relative entropy satisfies Stein's Lemma \cite{hiai_proper_1991,
ogawa_strong_2000}:
\begin{equation}
\forall \epsilon\in (0,1): \qquad     \lim_{n \to \infty} {\frac 1 n} D_H^ε(ρ^{\otimes n}\|σ^{\otimes n}) = D(ρ\|σ) ,
\end{equation}
meaning that for any fixed type-I error constraint $\epsilon$, the type-II error probability decays exponentially in the number of samples $n$, and the optimal error exponent is given by the quantum relative entropy $D(\rho \| \sigma)$. Additionally, the following one-shot bounds hold for all $α \in (0, 1)$, $α' >1$ and $ε \in (0,1)$  (see e.g.\ \cite{khatriPrinciplesQuantumCommunication2024}):
\begin{equation}\label{eq:hypothesis_testing_to_renyi}
     D_α(ρ\|σ) + {α \over α - 1} \log\left(1 \over ε\right) \leq D_H^ε(ρ\|σ) \leq \widetilde{D}_{α'}(ρ\|σ) + {α' \over α' - 1} \log\left(1 \over 1- ε\right)\,.
\end{equation}

In the symmetric setting of hypothesis testing, we assume that both hypotheses can occur with equal (prior) probability and then try to minimize the expected decision error probability (one can show that the asymptotic decay rate of this error (with $n$) is actually independent of the exact form of the prior as long as the prior involves both hypotheses with non-zero probability and stays constant in $n$). For state discrimination, this corresponds to the following expression:
\begin{align}
    \perr(ρ, σ) \coloneqq \frac{1}{2} \inf_{0\leq M\leq \iden_{\Hil}} (\alpha + \beta) &= {1 \over 2} \inf_{0 \leq M \leq \iden_{\Hil}} \big( \Tr((\iden_{\Hil} - M) \rho) + \Tr( M\sigma ) \big) \\ 
    &= {1 \over 2} \left(1 - {1 \over 2} \norm{ρ - σ}_1 \right),
\end{align}
where the last equality is known as the Holevo-Helström theorem \cite{helstrom_quantum_1969, holevo_analogue_1972}. This error probability also decays exponentially in $n$ when one considers multiple copies of $ρ$ and $σ$, and the optimal decay rate is given by the quantum Chernoff divergence \cite{audenaert_discriminating_2007, nussbaum_chernoff_2009}:
\begin{equation}
    \lim_{n \to \infty} - {1 \over n} \log \perr(ρ^{\otimes n}, σ^{\otimes n}) = \sup_{α \in (0,1)} (1 - α) D_α(ρ\|σ) \eqqcolon \xi(ρ\|σ)\,.
\end{equation}

\subsection{Quantum Channel Discrimination}
Quantum Channel Discrimination is the natural generalization of quantum state discrimination to the setting where the objects to be identified are quantum channels. Concretely, the task of \emph{binary quantum channel discrimination} is as follows: Given an unknown quantum channel as a black box and the side information that it is one of two possible channels, determine the channel's identity \cite{chiribella_memory_2008,duan_perfect_2009,hayashi_discrimination_2009,harrow_adaptive_2010}. This is again a fundamental building block in quantum information theory and has various applications in testing and verification of quantum devices and communication links. Nonetheless, many fundamental aspects of quantum channel discrimination are still not fully understood.
\subsubsection{Adaptive and Parallel Strategies}
Quantum channel discrimination is substantially more complex than quantum state discrimination due to the fact that one has the additional freedom of choosing input states, and if multiple uses of the channel are allowed, this choice can involve adaptivity and/or entanglement. Characterizing the optimal discrimination strategy and choice of input state(s) can thus be very tricky. For quantum channel discrimination, one often considers two classes of discrimination strategies: general \emph{adaptive} strategies, or strategies without adaptivity -- the so-called \emph{parallel} strategies. 

\begin{itemize}
    \item In a parallel strategy the input state to all the channels is fixed at the beginning and does not (adaptively) depend on channel outputs, however it can still involve states that are arbitrarily entangled. A parallel strategy is depicted in \autoref{fig:simple_parallel_strategy}, and it is specified by a joint input state $\nu_{R A^n} \in \mathcal{D}(RA^n)$. On the output side, we are then left with the state $(\id_R \otimes \Phi^{\otimes n})(\nu_{RA^n})$ or $(\id_R \otimes \Psi^{\otimes n})(\nu_{RA^n})$ depending on whether the given channel is $\Phi$ or $\Psi$. 
    \item In an adaptive strategy, the inputs are allowed to depend on previous channel outputs and also a (in principle unbounded) reference system which can be thought of as a `memory' for the adaptive computation. See \autoref{fig:simple_adaptive_strategy} for a depiction of an adaptive strategy. We call an adaptive strategy $\Theta$ that uses the unknown channel $n$ times an \emph{adaptive $n$-strategy}, and label with $\Theta[\Phi]$ (resp. $\Theta[\Psi]$) the joint state of the output and memory registers after the $n$th use of the channel $\Phi$ (resp. $\Psi$) (i.e.\ the joint state of all registers before the measurement $M$ in \autoref{fig:simple_adaptive_strategy}). It is easy to see from the figure that every adaptive $n$-strategy can be decomposed as $\Theta[\Phi] = (\id_R \otimes \Phi)(\Lambda_n(\Theta_{n - 1}[\Phi]))$, where $\Theta_{n-1}$ is an adaptive $(n-1)$-strategy, and $Λ_n$ is the last input-preparation channel of the strategy.
\end{itemize}

\begin{figure}[htb!]
    \centering
    \includegraphics[width=.4\linewidth]{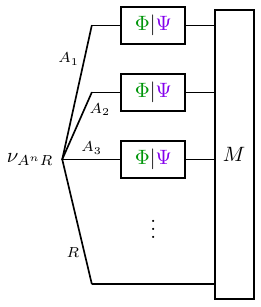}
    \caption{Illustration of a general parallel strategy with $n$ uses of the black-box channel, and a joint binary POVM measurement $\{M, \iden - M\}$ at the end. The strategy is fully specified by a joint input state $\nu$.}
\label{fig:simple_parallel_strategy}
\end{figure}

\begin{figure}[htb!]
    \centering
    \includegraphics[width=\linewidth]{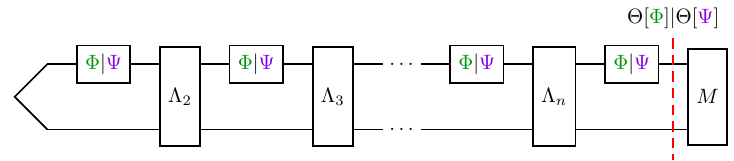}
    \caption{Illustration of a general adaptive strategy $\Theta$ with $n$ uses of the black-box channel, and a joint binary POVM measurement $\{M, \iden - M\}$ at the end. The top row uses the given black-boxes while the bottom row depicts the memory system $R$. The strategy is fully specified by an input state $\nu$ and intermediate input-preparation channels $\Lambda_2, \ldots ,\Lambda_n$. }
    \label{fig:simple_adaptive_strategy}
\end{figure}

It is not hard to see that every parallel strategy can be expressed as an adaptive strategy, however the reverse is not true. 

\subsubsection{Asymmetric Error Exponents} \label{sec:asymmetric-exponents}

For a discrimination strategy using the unknown channel $n$ times, let $\alpha_n$ and $\beta_n$ denote its type-I and type-II error probabilities, respectively. The \emph{$n$-shot type-II error exponent} is the normalised negative logarithm $-\frac{1}{n}\log\beta_n$. In asymmetric discrimination, we optimize this subject to a constraint on the type-I error, namely $\alpha_n\leq\epsilon$ for a fixed $\epsilon\in(0,1)$. The quantities $e_P$ and $e_A$ defined below are the optimal $n$-shot type-II error exponents for parallel and adaptive strategies, respectively.

The asymmetric error exponent is the only asymptotic error exponent for quantum channel discrimination for which one has been able to prove an entropic expression (although still involving a regularization limit). Concretely, for a parallel strategy as depicted in \autoref{fig:simple_parallel_strategy}, the $n$-shot type-II error exponent is characterized by the expression 
\begin{equation}
    \frac{1}{n} D_H^ε((\id_R \otimes \Phi^{\otimes n})(\nu_{RA^n})\|(\id_R \otimes \Psi^{\otimes n})(\nu_{RA^n}))\,,
\end{equation}
and for the optimal type-II parallel error exponent one would optimize over all input states:
\begin{equation}
    e_P(\Phi, \Psi, n , ε) = \sup_{\nu_{RA^n}} \frac{1}{n} D_H^ε((\id_R \otimes \Phi^{\otimes n})(\nu_{RA^n})\|(\id_R \otimes \Psi^{\otimes n})(\nu_{RA^n})). 
\end{equation}

One can show that the optimal asymptotic ($n\to \infty)$ type-II parallel error exponent as the type-I error goes to zero ($\epsilon\to 0$) is precisely the regularized cb-channel divergence \cite{wang_resource_2019}:
\begin{equation}
\begin{split}
    \lim_{\epsilon\to 0}\lim_{n\to \infty} e_P(\Phi, \Psi, n , ε) &=
    \lim_{ε \to 0} \lim_{n \to \infty} \sup_{\nu_{RA^n}} \frac{1}{n} D_H^ε((\id_R \otimes \Phi^{\otimes n})(\nu_{RA^n})\|(\id_R \otimes \Psi^{\otimes n})(\nu_{RA^n})) \\ &= \lim_{n \to \infty}\sup_{\nu_{RA^n}} {1 \over n} D((\id_R \otimes \Phi^{\otimes n})(\nu_{RA^n})\|(\id_R \otimes \Psi^{\otimes n})(\nu_{RA^n}))  \\ 
    &= D^{\cb, \reg}(\Phi\|\Psi). \label{eq:channel-steins-lemma}
    \end{split}
\end{equation}
While the right-hand side is an entropic expression only involving the two channels, there is no general way to compute it (in fact it is not even known whether it is theoretically computable), and explicit values are only known for very small classes of channels such as replacer channels, for which one can show that the regularization limit is not necessary \cite{wilde_amortized_2020}. 

Interestingly, even though we have no good way to compute this expression, it is possible to show that it is equal to the corresponding error exponent for adaptive discrimination strategies defined as:
\begin{equation}
    e_A (\Phi, \Psi, n, \varepsilon) := \sup_{\Theta \, \text{adaptive $n$-strategy}} \frac{1}{n} D_H^{\varepsilon} ( \Theta (\Phi) \| \Theta (\Psi) ),
\end{equation}
where for a given adaptive $n$-strategy $\Theta$, $\Theta[\cdot]$ denotes the joint state of all registers before the measurement $M$ in \autoref{fig:simple_adaptive_strategy} (see e.g.\ \cite{bergh_discrimination_2025} for an explicit expression of $e_A$ as an optimization problem involving the preparation maps and the input state that specifies $\Theta$). More precisely, the chain rule for the quantum relative entropy \cite{fang_chain_2020} implies that
\begin{equation}\label{eq:adaptive-channel-steins-lemma}
\lim_{ε \to 0} \lim_{n \to \infty} e_A(\Phi, \Psi, n , ε) = D^{\cb, \reg}(\Phi\|\Psi) = \lim_{\epsilon\to 0}\lim_{n\to \infty} e_P(\Phi, \Psi, n , ε).
\end{equation}
Therefore, adaptive strategies do not offer an asymptotically better asymmetric error exponent than parallel strategies. Besides these asymptotic expressions, the bounds from \eqref{eq:hypothesis_testing_to_renyi} can be turned into the following one-shot channel discrimination bounds that hold for all $α \in (0,1)$, $α' \in (1, \infty)$, and $ε \in (0,1)$ \cite{wang_resource_2019, wilde_amortized_2020, fawzi_defining_2021}:
\begin{align}
   e_P(\Phi, \Psi, n, ε)  & \geq  {1 \over n} D_α^{\cb}(\Phi\n \|\Psi\n) + {α \over n(α - 1)} \log\left(1 \over ε \right), \\
    e_A(\Phi, \Psi, n, ε) &\leq \sD^{\cb, \reg}_{\alpha'}(\Phi \| \Psi) + { α' \over n (α' - 1)} \log\left(1 \over 1 - ε\right).
\end{align}
In particular, we get the following bounds:
\begin{equation}\label{eq:channel_discrimination_n_shot_exponents_renyi}
\begin{split}
    D^{\cb}_α(\Phi \|\Psi) + {α \over n(α - 1)} \log \left(1 \over ε\right) &\leq e_P(\Phi, \Psi, n, ε) \\&\leq e_A(\Phi, \Psi, n, ε) \\ &\leq \widetilde{D}^{\cb, \reg}_{α'}(\Phi \| \Psi) + { α' \over n (α' - 1)} \log\left(1 \over 1 - ε\right).
\end{split}
\end{equation}

\subsubsection{Strong Converse Exponents} \label{sec:sc-exponents}
While the Channel Stein's Lemma \eqref{eq:adaptive-channel-steins-lemma} gives the optimal achievable error exponent for the channel discrimination task, it has currently only been proven with a weak converse, meaning that there is no discrimination strategy that achieves a better asymptotic type-II error exponent with asymptotically vanishing type-I error. Many other (similar) information theoretic tasks also satisfy the strong-converse property, which (for this discrimination problem) means that the optimal asymptotic type-II error exponent does not increase even if one allows a non-vanishing type-I error probability strictly bounded away from $1$. Concretely, this would mean that one can replace the limit $ε \to 0$ with a limit $ε \to 1$ in \eqref{eq:channel-steins-lemma} and the right-hand side stays unchanged (in fact, this would imply that one need not take a limit in $ε$ at all and the expression is the same for all $ε \in (0,1)$). 

The works \cite{wilde_amortized_2020, fawzi_defining_2021} have established that the optimal asymptotic type-II error exponent (for both adaptive and parallel strategies) with the type-I error probability just bounded away from $1$ is bounded from above by 
\begin{align}
\lim_{ε \to 1} &\limsup_{n \to \infty} e_{A/P}(\Phi, \Psi, n, ε) \\ &\leq 
\lim_{α \to 1^+} \widetilde{D}^{\cb, \reg}_α(\Phi\|\Psi) = \lim_{α \to 1} \lim_{n \to \infty} \sup_{\nu_{RA^n}} {1 \over n} \widetilde{D}_α((\id_R \otimes \Phi^{\otimes n})(\nu_{RA^n})\|(\id_R \otimes \Psi^{\otimes n})(\nu_{RA^n})).
\end{align}
If one can exchange the limits $α \to 1$ and $n \to \infty$ in the last expression, one would be able to show that the right-hand side is equal to $D^{\cb,\reg}(\Phi\|\Psi)$, thus establishing the strong converse property for quantum channel discrimination. However, actually proving that these limits can be exchanged has not been possible yet.
While this was claimed in \cite{Fang_2025towards}, the argument turned out to be incorrect, and this led to the discovery of errors in the proof of the generalized quantum Stein's lemma \cite{berta_gap_2023}. While we now have correct proofs of the generalized quantum Stein's lemma \cite{hayashiGeneralizedQuantumSteins2025, Lami2025generalized} the new approaches use different techniques that do not seem to be directly applicable for the argument proposed in \cite{Fang_2025towards} and so the question for the strong-converse in channel discrimination remains open. Note that the strong-converse property does hold for channels for which one can show that the regularization is not necessary, i.e.\ for which $\widetilde{D}^{\cb}_α(\Phi^{\otimes n}\|\Psi^{\otimes n}) = n \widetilde{D}^{\cb}_α(\Phi\|\Psi)$.

Another way to state the strong converse property is that for any type-II error exponent larger than the optimal achievable exponent with vanishing type-I error (for channel discrimination this is given by $D^{\cb, \reg}(\Phi\|\Psi)$), the type-I error has to approach one in the asymptotic limit. One can then ask the question how fast this type-I error converges to 1. This convergence is typically exponential (in $n$), and the corresponding rate (or exponent) is called the strong-converse exponent. Formally, for $r > 0$ one writes
\begin{align}
    \overline{e}^{\mathrm{sc}}_{A/P}(\Phi, \Psi, r) &:= \limsup_{n \to \infty} \inf_{\text{A/P strategy}}  \left\{ - {1 \over n} \log(1 - α_n) \, \big|\, β_n \leq 2^{-rn}\right\}, \\
    \underline{e}^{\mathrm{sc}}_{A/P}(\Phi, \Psi, r) &:= \liminf_{n \to \infty} \inf_{\text{A/P strategy}}  \left\{ - {1 \over n} \log(1 - α_n) \, \big| \, β_n \leq 2^{-rn}\right\},
\end{align}
where $A/P$ indicates the use of adaptive or parallel strategies, and $α_n$/$β_n$ are the type-I/II error probabilities of the chosen strategy with $n$ channel uses. In \cite{fawzi_defining_2021}, it was shown that these exponents are equal for adaptive and parallel strategies and are given by 
\begin{equation}\label{eq:fawzi-sc}
    e_A^{\mathrm{sc}}(\Phi, \Psi, r) = e_P^{\mathrm{sc}}(\Phi, \Psi, r) = \sup_{α > 1} {α - 1 \over α}(r - \widetilde{D}_α^{\cb, \reg}\big(\Phi\|\Psi)\big) ,
\end{equation}
where additionally\footnote{The papers \cite{fawzi_defining_2021, wilde_amortized_2020} only consider $\overline{e}^{\mathrm{sc}}_{A/P}$, however it is easy to see that \cite[Prop. 20]{wilde_amortized_2020} also works for $\underline{e}^{\mathrm{sc}}_{A}$, which together with the results from \cite{fawzi_defining_2021} implies that the two limits are equal.} $\overline{e}^{\mathrm{sc}}_{A/P} = \underline{e}^{\mathrm{sc}}_{A/P} \eqqcolon e^{\mathrm{sc}}_{A/P}$. 

Note that the strong converse exponent is non-zero if and only if $r> \inf_{\alpha>1} \sD^{\cb, \reg}_{\alpha}( \Phi \Vert \Psi)$, and it is not known whether 
\begin{equation}
    \inf_{\alpha>1} \sD^{\cb , \reg}_{\alpha} (\Phi \| \Psi) \stackrel{?}{=} D^{\cb, \reg}(\Phi \| \Psi),
\end{equation}
which corresponds precisely to the problem of establishing whether the (exponential) strong converse property holds in quantum channel discrimination \cite{fawzi_defining_2021, Fang_2025towards}.

\subsubsection{Symmetric Error Exponents} \label{sec:symmetric-exponents}

In the symmetric setting with equal (prior) probability, if one is allowed to use the channel only once, we can easily show that the smallest expected error probability is given by

\begin{align}
    \perr(\Phi, \Psi) &\coloneqq \inf_{\nu \in \mathcal{D}(RA)} \inf_{0 \leq M \leq \iden_{\Hil}} {1 \over 2} \big[ \Tr((\iden_{\Hil} - M)\Phi(\nu)) + \Tr(M\Psi(\nu)) \big] \\ 
    &= {1 \over 2}\left(1 - {1 \over 2} \norm{\Phi - \Psi}_{\diamond} \right),
\end{align}
where $\norm{\cdot}_{\diamond}$ is the diamond norm \cite[Chapter 3]{watrous_theory_2018}.
Similarly, for multiple uses of the channel one can again think about adaptive and parallel strategies, and arrive at the following expressions for the error probabilities:
\begin{align}
\perr^P(\Phi, \Psi, n) &\coloneqq \perr(\Phi\n, \Psi\n), \\
\perr^A(\Phi, \Psi, n) &\coloneqq \inf_{\Theta \text{ adaptive $n$-strategy}} {1 \over 2} \left(1 - {1 \over 2}\norm{\Theta[\Phi] - \Theta[\Psi]}_1 \right).
\end{align}

Since every parallel strategy is an adaptive strategy it is easy to see that 
\begin{equation}
    \perr^A(\Phi, \Psi, n) \leq \perr^P(\Phi\n, \Psi\n).
\end{equation}
Moreover, in contrast to the asymmetric setting \eqref{eq:adaptive-channel-steins-lemma}, there can be a strict separation between the optimal asymptotic parallel and adaptive error exponents in the symmetric setting, i.e.\ there exist channels $\Phi, \Psi$ such that \cite{harrow_adaptive_2010, salek_usefulness_2022}
\begin{equation}
    \limsup_{n\to \infty} -\frac{1}{n} \log (\perr^P (\Phi^{\otimes n}, \Psi^{\otimes n}) ) < \liminf_{n\to \infty} -\frac{1}{n} \log (\perr^A (\Phi, \Psi, n) ).
\end{equation}

There are currently no matching upper and lower bounds for these symmetric error exponents, although using the established results for symmetric state discrimination one can obtain the following two bounds.

\begin{lemma}[{\cite[Proposition 21]{wilde_amortized_2020}}] \label{lemma:perr<=}
For two quantum channels $\Phi, \Psi: \B{A} \to \B{B}$:
    \begin{equation}\limsup_{n \to \infty} - {1 \over n}\log (\perr^{A}(\Phi, \Psi, n)) \leq D^{\cb}_{\max}(\Phi\|\Psi).
    \end{equation}
\end{lemma}

\begin{lemma}\label{lemma:perr>=}
For two quantum channels $\Phi, \Psi: \B{A} \to \B{B}$:
    \begin{align}
    \liminf_{n \to \infty} &- {1 \over n} \log(\perr^{P}(\Phi\n, \Psi\n)) \geq \sup_{α \in (0, 1)} (1 - α) D_α^{\cb, \reg}(\Phi\|\Psi) \\
    &= \sup_{α \in (0, 1)} (1 - α) \lim_{n \to \infty} {1 \over n} \sup_{\nu \in \mathcal{D}(RA^n)}  D_α((\id_R \otimes\Phi\n)(\nu)\|(\id_R \otimes \Psi\n)(\nu))\\
    &= \xi^{\cb, \reg}(\Phi\|\Psi)
    \end{align}
    where $\xi^{\cb, \reg}(\Phi\|\Psi)$ is the regularized cb-channel divergence associated to the Chernoff divergence $\xi(ρ\|σ) \coloneqq \sup_{α \in (0,1)} (1 - α) D_α(ρ\|σ)$.
\end{lemma}
\begin{proof}
    It is easy to see that the limit in $n$ on the right-hand side is increasing. Hence, we can take the three suprema in any order, and let $α$, $n$ and $\nu = \nu_{AR^n}$ be chosen such that 
    \begin{equation}
        \sup_{α \in (0, 1)} (1 - α) D_α^{\cb, \reg}(\Phi\|\Psi) = (1 - α) {1 \over n} D_α((\id_R \otimes\Phi\n)(\nu)\|(\id_R \otimes \Psi\n)(\nu)) + δ
    \end{equation}
    for some $δ > 0$ arbitrarily small. From here on in this argument we will make the $\id_R$ implicit, i.e.\ we will write $\Phi\n(\nu)$ for $(\id_R \otimes \Phi\n)(\nu)$.
    Now consider the parallel discrimination strategy that, given $m$ copies of either $\Phi$ or $\Psi$ divides the $m$ copies into $k = k(m) = \floor{m/n}$ groups of $n$ channels, inputs the state $\nu$ into each group and just discards any remaining channels. We are then left with $k$ copies of either $\Phi\n(\nu)$ or $\Psi\n(\nu)$, and by the Chernoff bound for quantum state discrimination \cite{audenaert_discriminating_2007} we can discriminate these with error probability
    \begin{align}
        - \log \perr((\Phi\n(\nu))^{\otimes k}, (\Psi\n(\nu))^{\otimes k}) &\geq k (1 - α) D_α(\Phi\n(\nu)\|\Psi\n(\nu))+ o(k) \\&= n k \sup_{α \in (0, 1)} (1 - α) D_α^{\cb, \reg}(\Phi\|\Psi)  - kn δ + o(k)
    \end{align}
    Since $n$ is fixed, it is easy to see that
    \begin{equation}
    \lim_{m \to \infty} {\floor{m/n} n \over m} = 1
    \end{equation}
    and hence
    \begin{align}
        \liminf_{m \to \infty}  -{1 \over m}\log(\perr^{P}(\Phi^{\otimes m}, \Psi^{\otimes m})) 
        &\geq \liminf_{m \to \infty}  - {1 \over m}\log\perr((\Phi\n(\nu))^{\otimes k(m)}, (\Psi\n(\nu))^{\otimes k(m)}) \\
        &\geq \liminf_{m \to \infty} \left\{ {n k(m) \over m} \sup_{α \in (0, 1)} (1 - α) D_α^{\cb, \reg}(\Phi\|\Psi)  - {δ n k(m) \over m} + {o(k) \over m} \right\}\\
        &=  \sup_{α \in (0, 1)} (1 - α) D_α^{\cb, \reg}(\Phi\|\Psi) - δ
    \end{align}
    and $δ$ was arbitrarily small.
\end{proof}

\subsection{Multiplicative Domain}
For any unital completely positive (UCP) linear map $\Phi:\B{\Hil}\to \B{\Hil}$, the \emph{multiplicative domain} is defined as 
\begin{align}
    M_{\Phi} := \big\{ X\in \B{\Hil} \,\, \vert \,\, \forall Y\in \B{\Hil}: \Phi(XY) &= \Phi(X)\Phi(Y) \nonumber \\           
    \Phi(YX) &= \Phi(Y)\Phi(X)  
            \big\} .
\end{align}
Recall that any UCP linear map $\Phi$ satisfies the \emph{Schwarz inequality}: $\Phi(X^{\dagger}X)\geq \Phi(X^{\dagger})\Phi(X)$ for all $X\in \B{\Hil}$. Moreover, Choi showed that \cite{Choi1974schwarz} 
\begin{equation}\label{eq:choi-schwarz}
     M_{\Phi} = \{X\in \B{\Hil} \,\, \vert \,\, \Phi(X^{\dagger}X) = \Phi(X^{\dagger})\Phi(X), \,\,\Phi(XX^{\dagger}) = \Phi(X)\Phi(X^{\dagger}) \}.
\end{equation}

Below, we prove the rank non-decreasing property of unital quantum channels in two simple lemmas, which also appeared in \cite{Rahaman2020wielandt}. We provide short proofs for completeness.

\begin{lemma}\label{lemma:rank-Phi1}
    Let $\Phi:\B{\Hil}\to \B{\Hil}$ be a unital quantum channel. Then, 
    \begin{equation}
      \forall X\geq 0: \quad  \operatorname{rank}\Phi(X)\geq \operatorname{rank}X.
    \end{equation}
\end{lemma}
\begin{proof}
    Recall from Uhlmann's theorem that for Hermitian matrices $X,Y\in \B{\Hil}$, $\Phi(X)=Y$ for some unital quantum channel $\Phi \iff \lambda(X) \succ \lambda(Y)$, where $\lambda(X), \lambda(Y)\in \R^{\dim \Hil}$ are the vectors of eigenvalues (counted with multiplicity) of $X,Y$ respectively, and $\succ$ denotes the majorization preorder \cite[see Chap. 4]{watrous_theory_2018}. Now, let $X,Y$ be Hermitian matrices such that $Y=\Phi(X)$. Let 
    \begin{align}
        \lambda(X) &= (r_1,r_2, \ldots ,r_n, 0, \ldots,0), \\
        \lambda(Y) &= (s_1, s_2 \ldots ,s_m, 0, \ldots ,0),
    \end{align}
    where $n=\operatorname{rank}(X), m=\operatorname{rank}(Y)$, and the entries are arranged in a non-increasing order. We want to prove that $m\geq n$. Suppose this is not the case and $m<n$. Then, since $\Phi$ is trace preserving and $\lambda(X)\succ \lambda(Y)$, we get
    \begin{align}
        r_1 +\ldots + r_m+\ldots+r_n &= s_1 + \ldots +s_m \\
        r_1 + \ldots +r_m &\geq s_1 + \ldots +s_m,
    \end{align}
    which leads to a contradiction. Hence, $m\geq n$.
\end{proof}

\begin{lemma}\label{lemma:rank-Phi2}
    Let $\Phi:\B{\Hil}\to \B{\Hil}$ be a unital quantum channel. Then, for an orthogonal projection $P=P^{\dagger}=P^2$,
    \begin{equation}
        \operatorname{rank} (\Phi(P)) = \operatorname{rank}(P) \iff P\in M_{\Phi}.
    \end{equation}
\end{lemma}
\begin{proof}
    Let $P\in M_{\Phi}$. Then, $\Phi(P)=\Phi(PP)=\Phi(P)\Phi(P)$, so that $\Phi(P)$ is also a projection. Moreover, since $\Phi$ is trace-preserving, we obtain the desired assertion:
    \begin{equation}
     \operatorname{rank} (\Phi(P)) =  \Tr \Phi(P) = \Tr P = \operatorname{rank} (P).
    \end{equation}
    Conversely, suppose $\operatorname{rank}(\Phi(P))=\operatorname{rank}(P)$. Let $Q$ be the orthogonal projection onto $\im \Phi(P)$. Since $\Phi$ is unital, we get $\Phi(P)\leq \norm{\Phi(\iden_{\Hil})}_{\infty}Q = Q$ \cite{Russo1966}. Since $\Phi$ is trace-preserving, we get $\Tr \Phi(P)=\Tr P = \operatorname{rank}P= \operatorname{rank}(\Phi (P)) = \operatorname{rank}Q=\Tr Q$, implying that $\Phi(P)=Q$ is a projection. Hence,
    \begin{equation}
        \Phi(P) = \Phi(P)\Phi(P) \leq \Phi(PP) = \Phi(P),
    \end{equation}
    which, according to Eq.~\eqref{eq:choi-schwarz}, implies that $P\in M_{\Phi}$.
\end{proof}

\section{Main results}\label{sec:main_results}
We start this section with a brief explanation of the (well-known) algebraic structure of idempotent channels that is essential for phrasing our main results.

Let $\cP:\B{\Hil}\to \B{\Hil}$ be an idempotent quantum channel ($\cP=\cP^2$, where $\cP^2 = \cP \circ \cP$) such that $\cP(\iden_{\Hil})$ has full rank\footnote{Note that this is equivalent to saying that $\cP$ admits a full-rank invariant state $\delta=\cP(\delta)$.}. Then, the image of the adjoint ${\cP}^*:\B{\Hil}\to \B{\Hil}$ is a unital $*-$subalgebra\footnote{This means that $\im(\cP^*)\subseteq \B{\Hil}$ is a vector subspace such that $\iden_{\Hil}\in \im(\cP^*)$ and 
\begin{equation}\forall X,Y\in \im(\cP^*): \quad X^{\dagger}, XY\in \im (\cP^*). \end{equation} } $\im (\cP^*)\subseteq \B{\Hil}$ \cite{Lindblad1999fixed}. Since $\Hil$ is finite-dimensional, $\im (\cP^*)\cong \oplus_{k=1}^K \M{d_k}$ in the sense of $*-$algebra isomorphism \cite[Theorem 11.2]{Takesaki1979algebra}, and is canonically represented on $\Hil$. Hence, there exists an orthogonal decomposition $\Hil = \oplus_{k=1}^K D_k\otimes C_k$ with $\dim D_k=d_k$ such that \cite[Theorem 11.9]{Takesaki1979algebra} (see also Appendix~\ref{appen:three-layer}):
\begin{align}\label{eq:phasespace}
       \im (\cP^*) &= \bigoplus_{k} (\B{D_k}\otimes \iden_{C_k}).
\end{align}

In fact, the channel acts like a conditional expectation \cite[Proposition 1.5]{wolf_quantum_2012}:
\begin{align} 
\forall X\in \B{\Hil}: \quad
   {\cP}(X) &= \sum_k V_k \left( \Tr_{C_k} (V_k^{\dagger} X V_k) \otimes \delta_k \right) V_k^{\dagger}, \label{eq:phaseproj-1} \\
   {\cP}^*(X) &= \sum_k V_k \left( \Tr_{C_k} \left( (V^{\dagger}_k X V_k) (\iden_{D_k} \otimes \delta_{k}) \right) \otimes \iden_{C_k} \right) V_k^{\dagger}, \label{eq:phaseproj-2}
\end{align}
where $V_k: D_k \otimes C_k\to \Hil$ are the canonical inclusion isometries, $\Tr_{C_k}$ is the partial trace over $C_k$, and $\delta_k\in \State{C_k}$ are full rank-states. Let $\delta := \oplus_k (\iden_{D_k} \otimes d_{C_k} \delta_{k})$. Clearly, $\delta$ has full rank. Moreover, one can check
\begin{align}\label{eq:P-PTr}
    \cP(X) = \sqrt{\delta} \cP_{\Tr} (X) \sqrt{\delta} 
\end{align}
where $\cP_{\Tr}=\cP_{\Tr}^*$ is the unique trace-preserving conditional expectation onto $\im (\cP^*)$. Note that Eqs.~\eqref{eq:phaseproj-1}, \eqref{eq:phaseproj-2} can alternatively be written in terms of a direct sum \cite{Fukuda2007direct}:
\begin{align}
    \cP &= \bigoplus_k \id_{D_k} \otimes \cR_{\delta_k} \label{eq:P-directsum} \\
    \cP^* &= \bigoplus_k \id_{D_k} \otimes \cR^*_{\delta_k}, \label{eq:P*-directsum}
\end{align}
where $\id_{D_k}$ is the identity channel and $\cR_{\delta_k}(\cdot) = \Tr (\cdot) \delta_{k}$ is a replacer channel. 

\begin{remark}
    In terms of block matrices, the direct sum $\Phi \oplus \Psi : \B{A \oplus C} \to \B{B \oplus D}$ of two channels $\Phi:\B{A}\to \B{B}, \Psi: \B{C}\to \B{D}$ can be defined as
\begin{align}
\forall X\in \B{A}, \forall Y\in \B{C}: \qquad    \Phi \oplus \Psi\left(
    \begin{bmatrix} X & * \\ * & Y \end{bmatrix} \right) := \begin{bmatrix} \Phi(X) & 0 \\ 0 & \Psi(Y) \end{bmatrix}.
\end{align}
\end{remark}

The rigid algebraic structure of idempotent channels makes them amenable to study. The communication capacities of such channels were recently characterized in \cite{Singh2026markovian} (see also \cite{Singh2025thesis}, \cite{fawzi2025capacities}, \cite{Singh2025zero}), and the problem of simulating such channels from one another was recently studied in \cite{Delsol2025emulation}. For detailed proofs of the structure of idempotent channels stated in this section, see \cite[Chapter 6]{wolf_quantum_2012}, \cite[Theorem 2.7.3]{Singh2025thesis} or \cite[Proposition 1]{Delsol2025emulation}.

\subsection{Noiseless vs Noisy: Identity Against an Idempotent Channel}

In this section, we completely characterize the optimal error-exponents for discriminating between the noiseless identity channel and an arbitrary idempotent channel.

\begin{theorem} \label{theorem:IDvQ}
    Let $\mathcal{Q}:\B{\Hil} \to \B{\Hil}$ be an idempotent channel with $\cQ(\iden_{\Hil})$ of full rank, i.e., $\cQ=\oplus_{l=1}^L \id_{A_l}\otimes \mathcal{R}_{\omega_l}$, where $\Hil = \oplus_{l=1}^L A_l\otimes B_l$ and $\omega_l\in  \cD(B_l)$ are states with full rank. Then,
    \begin{equation}
     D_{\min}(\id_{\Hil} \Vert \cQ) =D(\id_{\Hil} \Vert \cQ ) =  D_{\max}(\id_{\Hil} \Vert \cQ) = \log \sum_l \Tr_{\min(d_{A_l}, d_{B_l})} (\omega_l^{-1}),
    \end{equation}
     where for $k\in \mathbb{N}, \Tr_k (X) = \lambda_{1}(X) + \ldots + \lambda_k(X)$ denotes the sum of the largest $k$ eigenvalues of a Hermitian matrix $X$. Consequently,
    \begin{align}
     D_{\min}^{\cb}(\id_{\Hil} \Vert \cQ) = D^{\cb}(\id_{\Hil} \Vert  \cQ) &=  D^{\cb}_{\max}(\id_{\Hil} \Vert \cQ)
     = \log \sum_l \Tr (\omega_l^{-1}),
    \end{align}
    and all the cb-divergences are additive:
    \begin{align}
     \forall n\in \mathbb{N}: \quad    D^{\cb}(\id_{\Hil}^{\otimes n} \Vert  \cQ^{\otimes n}) 
     = n D^{\cb}(\id_{\Hil} \Vert  \cQ). 
    \end{align}
\end{theorem}
\begin{proof}
Let us first tackle the converse bound. Consider the following calculation for the case of a single summand
    \begin{align}
        D_{\max}(\id_{AB} \Vert \id_A \otimes \mathcal{R}_{\omega_B}) &= \sup_{\rho} D_{\max}(\rho_{AB} \Vert \rho_A \otimes \omega_B)  \\
        &= \sup_{\psi} D_{\max} (\psi_{AB} \Vert \psi_A \otimes \omega_B)  \\
        &= \sup_{\psi} \log \norm{ (\psi^{-1/2}_A \otimes \omega^{-1/2}_B) \ket{\psi}\bra{\psi}_{AB} (\psi_A^{-1/2} \otimes \omega^{-1/2}_B )}_{\infty}  \\
        &\leq \sup_{\substack{\Pi=\Pi^2 = \Pi^{\dagger} \\ \Tr \Pi \leq \min (d_A,d_B) }} \log \Tr (\Pi \omega^{-1})  \\
        &= \log \Tr_{\min (d_A, d_B)} (\omega^{-1}), \label{eq:ky-fan}
    \end{align}
where we restricted the optimization to pure states because of the quasi-convexity of $D_{\max}$. To obtain the inequality, we used the fact that for any $\ket{\psi}_{AB} = \sum_i \sqrt \mu_i \ket{\alpha_i}_A \ket{\beta_i}_B$ (Schmidt decomposition), we have $\psi^{-1/2}_A \ket{\psi}_{AB} = \sum_i \ket{\alpha_i}_A \ket{\beta_i}_B$, so that 
\begin{equation}
    \norm{ (\psi^{-1/2}_A \otimes \omega^{-1/2}_B) \ket{\psi}\bra{\psi}_{AB} (\psi_A^{-1/2} \otimes \omega^{-1/2}_B) }_{\infty} = \sum_i \langle \beta_i | \omega^{-1} | \beta_i \rangle = \Tr (\Pi \omega^{-1}),
\end{equation}
where $\Pi=\sum_i \ket{\beta_i}\bra{\beta_i}$ is an orthogonal projection with $\Tr \Pi \leq \operatorname{SR}(\psi) \leq \min (d_A, d_B)$, where $\operatorname{SR}(\cdot)$ denotes the Schmidt rank. The final equality in Eq.~\eqref{eq:ky-fan} follows from Ky Fan's maximum principle \cite[Problem I.6.15 and Exercise II.1.13]{bhatia_matrix_1997}.

For the general case, note that for any projective measurement\footnote{This means that each $P_l=P_l^2=P_l^{\dagger}$ is an orthogonal projection and $\sum_l P_l = \iden_{\Hil}$.} $(P_l)_{l} \subseteq \B{\Hil}$ and a positive semi-definite $X\in \B{\Hil}$, the generalized pinching inequality shows that \cite{winter2025pinching}
\begin{equation}
    X \leq \sum_{l} \alpha_l P_l X P_l
\end{equation}
for any $\vec\alpha\in \R^{L}$ satisfying $\operatorname{diag}(\vec\alpha) \geq J$, where $J\in \mathbb{M}_L({\mathbb{R}})$ is the all-ones matrix. Define $t_l = \Tr_{\min (d_{A_l}, d_{B_l})} (\omega^{-1}_{l})$, the sum $T=\sum_l t_l$ and let $\alpha_l = T/t_l$. This $\vec \alpha$ satisfies the required constraint $\operatorname{diag}(\vec\alpha) \geq J$, since for any vector $\ket{v}\in \C^{L}$,
\begin{equation}
    \bra{v} J \ket{v} = \abs{\sum_l v_l}^2 = \abs{ \sum_l \frac{v_l}{\sqrt{t_l}} \sqrt{t_l} }^2 \leq \sum_l \frac{\abs{v_l}^2}{t_l} \sum_l t_l = \sum_l \frac{T \abs{v_l}^2}{t_l} = \bra{v} \operatorname{diag}(\vec \alpha) \ket{v}.
\end{equation}
Hence, choosing $P_l=V_lV_l^{\dagger}$ to be the orthogonal projection onto $A_l \otimes B_l \subseteq \Hil$ with $V_l : A_l\otimes B_l \to \Hil$ being the canonical inclusion isometry, generalized pinching shows
\begin{equation}
\forall \rho\in \cD(\Hil) : \quad    \rho \leq T \sum_l \frac{1}{t_l} P_l \rho P_l \leq T \sum_l V_l \left( \Tr_{B_l}(V^{\dagger}_l \rho V_l) \otimes \omega_l \right) V_l^{\dagger} = T \cQ (\rho),
\end{equation}
where Eq.~\eqref{eq:ky-fan} and the operator-order characterization of $D_{\max}$ (see Remark~\ref{remark:Dmax-operator-order}) give $X_l\leq t_l\bigl(\Tr_{B_l}(X_l)\otimes\omega_l\bigr)$ for $X_l=V_l^{\dagger}\rho V_l$; and conjugating by $V_l$, dividing by $t_l$, and summing over $l$ yields the second inequality. 
Hence, 
\begin{align}
    D_{\max}(\id \Vert \oplus_l \id_{A_l} \otimes \cR_{\omega_l}) \leq \log T = \log \sum_l \Tr_{\min (d_{A_l}, d_{B_l})} (\omega^{-1}_l) .
\end{align}
To show achievability, consider the following calculation for the case of a single summand:
\begin{align}
    \sup_{\psi}D(\psi_{AB} \Vert \psi_A \otimes \omega_B ) &= \sup_{\psi} - \Tr (\psi_{AB} \log (\psi_A \otimes \omega_B)) \\
    &= \sup_{\psi} - \Tr (\psi_{AB} (\log \psi_A \otimes \iden_B + \iden_A\otimes \log \omega_B )) \\ 
    &= \sup_{\psi} -\Tr (\psi_A \log \psi_A) - \Tr (\psi_B \log\omega_B) \\ 
    &\overset{(a)}{\geq} \max_{\{\mu_i \}_{i}} H(\{ \mu_i \}) + \sum_i \mu_i \log (\lambda_i (\omega^{-1})) \\
    & \overset{(b)}{=} \log \sum_{i} \lambda_i(\omega^{-1}) \\
    &= \log \Tr_{\min(d_A,d_B)} (\omega^{-1})
\end{align}
where the middle inequality $(a)$ follows by choosing $\ket{\psi}_{AB}=\sum_i \sqrt{\mu_i} \ket{\alpha_i}_A\ket{\beta_i}_B$ to be of full Schmidt rank such that $\ket{\beta_i}_B$ are precisely the eigenvectors of $\omega^{-1}_B$ associated with the largest $\min(d_A,d_B)$ eigenvalues, and the equality in $(b)$ follows from maximizing a function of the form $H(\{\mu_i\}_i) + \sum_i \mu_i c_i$ over the simplex of probability distributions $\{\mu_i\}_i$:
\begin{equation}
    \max_{\mu_i\geq 0, \sum_i \mu_i =1 } \left( H(\{\mu_i \}_i) + \sum_i \mu_i c_i \right) = \log \left( \sum_i 2^{c_i} \right),
\end{equation}
see Appendix~\ref{appen:aux}. This proves that
\begin{equation}
    D(\id_{AB} \Vert \id_A \otimes \cR_{\omega_B}) = D_{\max}(\id_{AB} \Vert \id_A \otimes \cR_{\omega_B}).
\end{equation}
For $D_{\min}(\cdot \Vert \cdot)$, we can do a similar calculation:
\begin{align}
    D_{\min}(\id_{AB} \Vert \id_A \otimes \cR_{\omega_B}) &= \sup_{\rho} D_{\min}(\rho_{AB}\Vert \rho_A \otimes \omega_B) \\
    &\geq \sup_{\psi}D_{\min}(\psi_{AB} \Vert \psi_A \otimes \omega_B) \\ 
    &=\sup_{\psi} -\log \Tr (\psi_{AB} (\psi_A \otimes \omega_B)) \\
    &\overset{(a)}{\geq}\max_{\{\mu_i\}_i} -\log \sum_i \mu_i^2 \frac{1}{\lambda_i (\omega^{-1})} \\
    &\overset{(b)}{=} -\log \frac{1}{\sum_i \lambda_i (\omega^{-1}) } \\
    &= \log \Tr_{\min(d_A,d_B)} (\omega^{-1}),
\end{align}
where the middle inequality $(a)$ follows by choosing $\ket{\psi}_{AB}=\sum_i \sqrt{\mu_i} \ket{\alpha_i}_A\ket{\beta_i}_B$ to be of full Schmidt rank such that the Schmidt vectors $\ket{\beta_i}_B$ are precisely the eigenvectors of $\omega^{-1}_B$ associated with the largest $\min(d_A,d_B)$ eigenvalues, and the equality in $(b)$ follows by minimizing a function of the form $\sum_i \mu_i^2 c_i$ for non-negative numbers $\{c_i\}_i$ over the simplex of probability distributions $\{\mu_i \}_i$ (see Appendix~\ref{appen:aux}):
\begin{equation}
    \min_{\mu_i \geq0, \sum_i\mu_i=1} \left( \sum_i \mu_i^2 c_i \right) = \frac{1}{\sum_i 1/c_i}.
\end{equation}

The full achievability result follows from a similar argument. Consider $\Hil = \oplus_l A_l \otimes B_l$ and let $\ket{\psi} = \oplus_l \sqrt{p_l} \sum_i \sqrt{\mu_{i,l}} \ket{\alpha_{i,l}}_{A_l}\ket{\beta_{i,l}}_{B_l}$, where $p_l\geq 0$, $\sum_lp_l=1$, and within each $l-$block, the vectors $\ket{\alpha_{i,l}}\in A_l$ and $\ket{\beta_{i,l}}\in B_l$ and the Schmidt coefficients $\mu_{i,l}$ are chosen as above. Let $W_l:A_l \otimes B_l \to \Hil$ be the canonical inclusion isometries, and $\cQ_l := \id_{A_l}\otimes \cR_{\omega_l}$. Then, 

\begin{align}
    D_{\min}(\psi \Vert \cQ (\psi)) &= -\log \Tr (\psi \cQ (\psi)) \label{eq:Dmin-start} \\
    &= -\log \sum_l \Tr (W_l^{\dagger}\psi W_l \cQ_l (W^{\dagger}_l \psi W_l)  ) \\
    &= -\log \sum_l p_l^2 \Tr \left( \frac{W_l^{\dagger}\psi W_l}{p_l} \cQ_l \left(\frac{W_l^{\dagger}\psi W_l}{p_l} \right) \right) \\
    &= -\log \sum_l p_l^2 \frac{1}{ \Tr_{\min(d_{A_l}, d_{B_l})}(\omega_{l}^{-1})} \\
    &= -\log \frac{1}{\sum_l \Tr_{\min(d_{A_l}, d_{B_l})}(\omega_l^{-1}) } \\
    &= \log \sum_l \Tr_{\min(d_{A_l}, d_{B_l})}(\omega_{l}^{-1}), \label{eq:Dmin-end}
\end{align}
where the final two equalities follow by choosing $p_l \propto \Tr_{\min(d_{A_l}, d_{B_l})}(\omega_{l}^{-1})$.

The expressions for the $\cb-$divergences follow easily:
\begin{align}
    D^{\cb}(\id_{\Hil} \Vert \cQ) &= D(\id_{\Hil} \otimes \id_{\Hil} \Vert \id_{\Hil} \otimes \cQ) \\
   &= D(\id_{\Hil \otimes\Hil} \Vert \oplus_l \id_{\Hil\otimes A_l} \otimes \cR_{\omega_l}) \\
   &= \log \sum_l \Tr_{\min(d_{\Hil}d_{A_l}, d_{B_l})}(\omega_{l}^{-1}) \\
   &=\log \sum_l \Tr(\omega_{l}^{-1}).
\end{align}
Finally, the additivity claim follows since the $\max-\cb$ divergence is always additive \cite{wilde_amortized_2020}. Alternatively, we can directly check that the stated formula is additive.
\end{proof}

\begin{corollary} \label{corollary:IdvQ}
    Let $\mathcal{Q}:\B{\Hil} \to \B{\Hil}$ be an idempotent channel with $\cQ(\iden_{\Hil})$ of full rank. Then, for any $\epsilon\in (0,1)$ and $r>D^{\cb}(\id_{\Hil}\Vert \cQ)$:
    \begin{align}
        \lim_{n\to \infty} e_P(\id_{\Hil},\cQ,n,\epsilon) = \lim_{n\to \infty} e_A(\id_{\Hil},\cQ,n,\epsilon) &= D^{\cb}(\id_{\Hil}\Vert \cQ) \\
        e^{\operatorname{sc}}_P(\id_{\Hil},\cQ, r) = e^{\operatorname{sc}}_A (\id_{\Hil},\cQ,r) &= r - D^{\cb}(\id_{\Hil} \Vert \cQ) \\
        \lim_{n\to \infty} -\frac{1}{n} \log (\perr^{P}(\id_{\Hil}, \cQ, n)) = \lim_{n\to \infty} -\frac{1}{n} \log (\perr^{A}(\id_{\Hil}, \cQ, n)) &= D^{\cb}(\id_{\Hil} \Vert \cQ).
    \end{align}
\end{corollary}

\begin{proof}Besides just the asymptotic expressions, we actually get the following slightly stronger $n$-shot bounds from \eqref{eq:channel_discrimination_n_shot_exponents_renyi}, for all $\alpha\in (0,1)$, $\alpha'>1$, and $\epsilon\in (0,1)$: 
    \begin{align}
     D^{\cb}(\id_{\Hil} \Vert \cQ) + \frac{\alpha}{n(\alpha-1)}\log \left( \frac{1}{\epsilon} \right) &\leq e_P(\id_{\Hil},\cQ,n,\epsilon) \\
     &\leq e_A(\id_{\Hil}, \cQ,n,\epsilon) \\
     &\leq D^{\cb}(\id_{\Hil} \Vert \cQ) + \frac{\alpha'}{n(\alpha'-1)} \log\left( \frac{1}{1-\epsilon} \right),
\end{align}
where we used Theorem~\ref{theorem:IDvQ} to deduce that $D_{\alpha}^{\cb}(\id_{\Hil} \Vert \cQ)=D^{\cb}(\id_{\Hil} \Vert \cQ) = \sD_{\alpha'}^{\cb}(\id_{\Hil} \Vert \cQ)$ for all $\alpha\in (0,1)$, $\alpha'>1$, as well as $D^{\cb}(\id_{\Hil}^{\otimes n} \Vert  \cQ^{\otimes n}) 
     = n D^{\cb}(\id_{\Hil} \Vert  \cQ)$. The first claim then follows by taking the limit $n\to \infty$.

The second claim follows from Eq.~\eqref{eq:fawzi-sc} and the additivity result from Theorem~\ref{theorem:IDvQ}:
\begin{align}
    e_A^{\mathrm{sc}}(\id_{\Hil}, \cQ, r) = e_P^{\mathrm{sc}}(\id_{\Hil}, \cQ, r) &= \sup_{α > 1} {α - 1 \over α}(r - \widetilde{D}_α^{\cb, \reg}\big(\id_{\Hil}\|\cQ)\big) \\
    &= r - D^{\cb}(\id_{\Hil} \Vert \cQ).
\end{align}

The final claim follows from Lemma~\ref{lemma:perr<=} and Lemma~\ref{lemma:perr>=}:
    \begin{align}
          \limsup_{n \to \infty} - {1 \over n}\log (\perr^{A}(\id_{\Hil}, \cQ, n)) &\leq D^{\cb}_{\max}(\id_{\Hil}\|\cQ) = D^{\cb}(\id_{\Hil}\Vert \cQ),
     \end{align}     
     \begin{align}
          \liminf_{n \to \infty} - {1 \over n} \log(\perr^{P}(\id_{\Hil}, \cQ, n)) &\geq \sup_{α \in (0, 1)} (1 - α) D_α^{\cb, \reg}(\id_{\Hil}\|\cQ) \\
          &= D^{\cb}(\id_{\Hil}\Vert \cQ),
    \end{align}
where $D^{\cb}_{\max}(\id_{\Hil}\|\cQ) = D^{\cb}(\id_{\Hil}\Vert \cQ)=D^{\cb, \reg}_{\alpha}(\id_{\Hil}\|\cQ)$ again follow from Theorem~\ref{theorem:IDvQ}.
\end{proof}

\subsection{Noisy vs Noisy: Two Idempotent Channels}

In this section, we attempt to characterize the optimal error-exponents for discriminating between two idempotent quantum channels $\cP,\cQ:\B{\Hil}\to \B{\Hil}$. Throughout this section, we assume that both $\cP(\iden_{\Hil})$ and $\cQ(\iden_{\Hil})$ have full rank.

We begin by noting a general lemma.

\begin{lemma}\label{lemma:D(PQ)=inf}
    Let $\cP, \cQ : \B{\Hil}\to \B{\Hil} $ be idempotent channels such that both $\cP(\iden_{\Hil}), \cQ(\iden_{\Hil})$ have full rank. If $\im (\cQ^*) \not\subseteq \im (\cP^*)$, then $D(\cP \Vert \cQ)= +\infty$.
\end{lemma}
\begin{proof}
    As noted already, there exist positive definite operators $\delta\in \im (\cP)$ and $\omega\in \im (\cQ)$ such that (see Eq.~\eqref{eq:P-PTr})
    \begin{equation}
        \cP(X) = \sqrt{\delta} \cP_{\Tr} (X) \sqrt{\delta}, \quad \cQ(X) = \sqrt{\omega} \cQ_{\Tr} (X) \sqrt{\omega},
    \end{equation}
    where $\cP_{\Tr}, \cQ_{\Tr}$ are the unique trace-preserving conditional expectations onto $\im (\cP^*), \im (\cQ^*)$ respectively, see Eq.~\eqref{eq:P-PTr}. Assume that $\im (\cQ^*) \not\subseteq \im (\cP^*)$. Since any $*-$subalgebra in $\B{\Hil}$ is generated by its orthogonal projections, there exists an orthogonal projection $Q\in \im (\cQ^*)$ such that $Q \not\in\im (\cP^*)$. Moreover, since $\cP_{\Tr}, \cQ_{\Tr}$ are unital quantum channels, Lemmas~\ref{lemma:rank-Phi1} and \ref{lemma:rank-Phi2} show that 
    \begin{equation}
        \operatorname{rank} (\cP(Q)) =\operatorname{rank} (\cP_{\Tr}(Q) )>\operatorname{rank}Q = \operatorname{rank} (\cQ_{\Tr}(Q))=\operatorname{rank} (\cQ(Q)).
    \end{equation}
    The inequality is strict because the multiplicative domain of the trace-preserving conditional expectation $\cP_{\Tr}$ is precisely its range $\im(\cP^*)$, so that $Q\notin M_{\cP_{\Tr}}$ and Lemma~\ref{lemma:rank-Phi2} rules out equality.
    Hence, by choosing $\rho = Q/\Tr Q$, we see $D(\cP\Vert \cQ) \geq D(\cP (\rho)\Vert \cQ(\rho)) = + \infty$.
\end{proof}

\begin{remark}
    Lemma~\ref{lemma:D(PQ)=inf} also gives a sufficient condition for $\sD_α(\cP\|\cQ) = \infty$ when $α > 1$, since $\sD_{\alpha}(\cP \Vert \cQ)\geq D(\cP \Vert \cQ)$. However, a similar statement is not true for $α < 1$. For $α < 1$, $D_α(\cP\|\cQ) = \infty$ if and only if there exists a state $ρ$ such that $\cP(ρ)$ and $\cQ(\rho)$ are orthogonal (to see this note that the set of states is compact, and for $α < 1$, $D_α$ is continuous). However, for idempotent channels $\cP, \cQ : \B{\Hil}\to \B{\Hil} $ with $\cP(\iden_{\Hil}), \cQ(\iden_{\Hil})$ of full rank, this never happens. To see this, note that \autoref{theorem:IDvQ} implies that for any state $\rho\in \cD (\Hil)$ the support projections satisfy $\Pi_{\rho}\leq \Pi_{\cP(\rho)}$ and $\Pi_{\rho} \leq \Pi_{\cQ(\rho)}$, and so $\cP(\rho)$ and $\cQ(\rho)$ cannot be orthogonal. 
\end{remark}

In accordance with Lemma~\ref{lemma:D(PQ)=inf}, we now restrict ourselves to discriminating between idempotent channels $\cP, \cQ : \B{\Hil}\to \B{\Hil} $ with $\im (\cQ^*) \subseteq \im (\cP^*) \subseteq \B{\Hil}$. Since $\Hil$ is finite-dimensional, these image algebras are $*$-isomorphic to direct sums of matrix blocks \cite[Theorem 11.2]{Takesaki1979algebra}: $\im (\cP^*) \cong \oplus_{k=1}^K \M{d_k}$ and $\im (\cQ^*) \cong \oplus_{l=1}^L \M{a_l}$, and there exist orthogonal direct sum decompositions of the underlying Hilbert space (see Appendix~\ref{appen:three-layer}):
\begin{align}
    \Hil &= \oplus_{k=1}^K \oplus_{l=1}^L A_l \otimes B_{k,l} \otimes C_k \label{eq:H=ABC1} \\
    &= \oplus_k \underbrace{(\oplus_l A_l \otimes B_{k,l})}_{:=D_k} \otimes C_k \label{eq:H=ABC2} \\
    &= \oplus_l A_l \otimes \underbrace{(\oplus_k B_{k,l} \otimes C_k)}_{:=E_l} \label{eq:H=ABC3}
\end{align}
such that $\dim{D_k}= d_k, \dim A_l = a_l, \dim B_{k,l} \geq 0$, $\dim C_k \geq 1$ and
\begin{align}
    \im (\cP^*) &= \bigoplus_k \B{D_k} \otimes \iden_{C_k}  \\
    \im (\cQ^*) &= \bigoplus_l \B{A_l} \otimes \iden_{E_l}.
\end{align}
Denote the orthogonal projection onto a single $k,l$ summand $A_l \otimes B_{k,l} \otimes C_k\subseteq \Hil$ by $\Pi_{k,l}$, so that the minimal central projections of $\im (\cP^*)$ and $\im (\cQ^*)$ become 
\begin{align}
    P_k = \sum_l \Pi_{k,l} \quad\text{and}\quad Q_l = \sum_k \Pi_{k,l},
\end{align}
which project onto $D_k\otimes C_k\subseteq \Hil$ and $A_l \otimes E_l\subseteq \Hil$, respectively. Note that $P_k Q_l = Q_l P_k = \Pi_{k,l}$ and for some $k,l$, the projection $\Pi_{k,l}$ might be zero, which corresponds to the fact that $\dim B_{k,l}=0$ for those $k,l$. We can now write the actions of $\cP, \cQ$ as direct sums with respect to the appropriate orthogonal decompositions (see Eq.~\eqref{eq:P-directsum}):
\begin{align}
    \cP = \bigoplus_k \id_{D_k} \otimes \cR_{\delta_k}, \qquad
    \cQ = \bigoplus_l \id_{A_l} \otimes \cR_{\omega_l}, \label{eq:Q-l}
\end{align}
where $\delta_k \in \cD(C_k)$ and $\omega_l \in \cD (E_l)$ are full-rank states. Alternatively, 
\begin{align}
    \cP (X) &= \sum_k V_k \cP_k (V_k^{\dagger} XV_k) V_k^{\dagger}, \\
    \cQ (X) &= \sum_l W_l \cQ_l (W_l^{\dagger} XW_l) W_l^{\dagger},
\end{align}
where $V_k:D_k \otimes C_k\to \Hil$ and $W_l:A_l\otimes E_l \to \Hil$ are the canonical isometries, and $\cP_k : \B{D_k \otimes C_k}\to \B{D_k \otimes C_k}$ and $\cQ_l : \B{A_l \otimes E_l}\to \B{A_l \otimes E_l}$ act as tensor products of identity and replacer: $\cP_k = \id_{D_k} \otimes \cR_{\delta_k}$ and $\cQ_l = \id_{A_l} \otimes \cR_{\omega_l}$. Since $\im(\cQ^*)\subseteq \im (\cP^*)$, 
\begin{equation}\label{eq:QP=Q}
   \cP^*\circ \cQ^* = \cQ^* \quad\text{and} \quad   \cQ \circ \cP = \cQ.
\end{equation}

In general, $\omega_l\in \cD (E_l)=\cD(\oplus_k B_{k,l}\otimes C_k)$ can have coherences across different $k$-blocks, and within each $k$-block, it can be entangled across the $B_{k,l}-C_k$ cut, which makes it trickier to compute the corresponding channel divergences $D(\cP \Vert \cQ)$ (see Section~\ref{sec:full-general}). Hence, in the next section, we first work with the assumption that $\cP$ and $\cQ$ share a common full-rank invariant state, which forces $\omega_l$ to be $k-$block diagonal and product across $B_{k,l}-C_k$.

\subsubsection{Common Invariant State}

With the setup of the previous section, we further assume that $\cP, \cQ$ share a common full-rank invariant state $\tau=\cP(\tau)=\cQ(\tau)$. Since any such $\tau$ has to be block diagonal with respect to both the $k$ and $l$ block decompositions (Eq.~\eqref{eq:Q-l}), the states $\omega_l\in \cD(E_l)= \cD (\oplus_k B_{k,l}\otimes C_k)$ must be $k-$block diagonal and product across the $B_{k,l}-C_k$ cut:  
\begin{equation}\label{eq:omega-block-prod}
  \omega_l = \bigoplus_k  p_{kl} \, \tau_{k,l} \otimes \delta_k, 
\end{equation}
where for each $l$, $p_{k,l}\geq 0, \sum_k p_{k,l}=1$, and $\tau_{k,l}\in \cD(B_{k,l})$ are full-rank states. Note that since $\omega_l$ has full rank, $p_{kl}>0$ whenever $\dim B_{k,l}>0$. 

We can now state the main result of this section.

\begin{theorem} \label{theorem:PvsQ:common}
    Let $\cP, \cQ : \B{\Hil}\to \B{\Hil} $ be idempotent channels with $\im (\cQ^*)\subseteq \im (\cP^*)$ that share a common full-rank invariant state. Then, $\Hil$ decomposes as 
    \begin{align}
    \Hil &= \oplus_{k=1}^K \oplus_{l=1}^L A_l \otimes B_{k,l} \otimes C_k \\
    &= \oplus_k \underbrace{(\oplus_l A_l \otimes B_{k,l})}_{:=D_k} \otimes C_k  \\
    &= \oplus_l A_l \otimes \underbrace{(\oplus_k B_{k,l} \otimes C_k)}_{:=E_l}, 
    \end{align}
    with 
    \begin{equation}
        \cP = \bigoplus_k \id_{D_k} \otimes \cR_{\delta_k}, \qquad \cQ=\bigoplus_l \id_{A_l}\otimes \cR_{\omega_l},
    \end{equation}
    where $\delta_k\in \cD(C_k)$ and $\omega_l\in \cD(E_l)$ are full-rank states such that 
    \begin{equation}
        \omega_l = \bigoplus_k p_{k,l} \tau_{k,l} \otimes \delta_k
    \end{equation}
    for some probability distributions $(p_{k,l})_k$ and full-rank states $\tau_{k,l}\in \cD(B_{k,l})$. Furthermore,
    \begin{align}
      D_{\min}(\cP \Vert \cQ) =  D(\cP\Vert \cQ) = D_{\max}(\cP\Vert \cQ) &= \max_k \left( \log \sum_l \frac{\Tr_{\min(d_{A_l}, d_{B_{k,l}})} (\tau_{k,l}^{-1})}{p_{k,l}} \right) \\
      D^{\cb}_{\min}(\cP \Vert \cQ) =  D^{\cb}(\cP\Vert \cQ) = D^{\cb}_{\max}(\cP\Vert \cQ) &= \max_k \left( \log \sum_l \frac{\Tr (\tau_{k,l}^{-1})}{p_{k,l}} \right). \label{eq:D-closed-expression}
    \end{align}
    Consequently, the $\cb$-divergences are all additive:
    \begin{align}
     \forall n\in \mathbb{N}: \quad   D^{\cb}(\cP^{\otimes n}\Vert \cQ^{\otimes n}) = n D^{\cb}(\cP \Vert \cQ).
    \end{align}
    Thus, for any quantum R\'enyi divergence $\mathbf{D}_α$ (i.e.\ any quantum divergence that satisfies the data-processing inequality and reduces to the classical R\'enyi divergence $D_α$ on commuting states), the following holds true:
    \begin{equation}
        \mathbf{D}_α^{\cb, \reg}(\cP\|\cQ) = \mathbf{D}_α^{\cb}(\cP\|\cQ) = D^{\cb}(\cP\|\cQ)\,.
    \end{equation}
\end{theorem}

\begin{proof} 
The proof of the first claim is described in the discussion that precedes the theorem (see also Appendix~\ref{appen:three-layer}). For the second claim, note that 
\begin{align}
    D_{\max}(\cP \Vert \cQ) &:= \sup_{\rho\in \cD(\Hil)} D_{\max}(\cP(\rho) \Vert \cQ(\rho) )  \\
    &= \sup_{\rho\in \cD(\Hil)} D_{\max}(\cP (\rho) \Vert \cQ \circ \cP (\rho) ) \\
    &= \sup_{\rho\in \im (\cP)} D_{\max}( \rho \Vert \cQ (\rho) ) \\
    &= \sup_{ \rho\in \im (\cP) \text{ extremal } } D_{\max} \left( \rho \big\Vert \cQ(\rho) \right),
\end{align}
where the equality $\cQ \circ \cP = \cQ$ follows from $\im (\cQ^*)\subseteq \im (\cP^*)$ and the final equality follows from the joint quasi-convexity of $D_{\max}$. Note that extremal states in $\im (\cP)=\oplus_k \B{D_k}\otimes \delta_k$ are of the form $\rho = 0\oplus \ldots\oplus (\psi \otimes \delta_k) \oplus \ldots \oplus 0$, supported on $D_k \otimes C_k$ for some $k$ and pure state $\psi\in \B{D_k}$. Recall that $D_k = \oplus_l A_l \otimes B_{k,l}$. Moreover, 
\begin{align}
    \cQ(\rho) = \bigoplus_l q_l \Tr_{B_{k,l}}(\psi_l) \otimes \omega_l = \bigoplus_{k'} \bigoplus_l p_{k'l} q_l \Tr_{B_{k,l}}(\psi_l)   \otimes \tau_{k',l} \otimes \delta_{k'}
\end{align}
where $\psi_l:= \frac{1}{q_l}V_{l|k}^{\dagger} \psi V_{l|k}\in \cD (A_l \otimes B_{k,l})$ and $q_l=\Tr (V_{l|k}^{\dagger} \psi V_{l|k})$, with $V_{l|k}:A_l\otimes B_{k,l}\to D_k$ being the canonical inclusion isometries. Hence,
\begin{align}
   \sup_{ \rho\in \im (\cP) \text{ pure } } D_{\max} ( \rho &\Vert \cQ(\rho) ) \nonumber \\
    &= \max_k \sup_{\psi\in \cD(D_k)}  D_{\max}(\psi\Vert \oplus_l p_{k,l} q_l \, \Tr_{B_{k,l}}(\psi_l)\otimes \tau_{k,l}) + D_{\max}(\delta_k\Vert \delta_k) \\
    &=  \max_k D_{\max} (\id_{D_k} \Vert \oplus_l p_{k,l} \, \id_{A_l} \otimes \cR_{\tau_{k,l}}) \\
    &= \max_k \left( \log \sum_l \frac{\Tr_{\min(d_{A_l}, d_{B_{k,l}})} (\tau_{k,l}^{-1})}{p_{k,l}} \right), \label{eq:max-k}
\end{align}
where the last equality follows from Theorem~\ref{theorem:IDvQ}. 

For the achievability claim, we can choose the block $k^*$ that achieves the maximum in Eq.~\eqref{eq:max-k} and let $\rho^*=0\oplus \ldots\oplus (\psi^* \otimes \delta_{k^*}) \oplus \ldots \oplus 0$ be supported on $D_{k*} \otimes C_{k^*}$ such that $\psi^*\in \cD(D_{k^*})$ is constructed as in the proof of Theorem~\ref{theorem:IDvQ} (see Eq.~\eqref{eq:Dmin-start}-\eqref{eq:Dmin-end}), so that
\begin{align}
    D_{\min} ( \rho^* \Vert \cQ(\rho^*) ) &=  \log \sum_l \frac{\Tr_{\min(d_{A_l}, d_{B_{k^*,l}})} (\tau_{k^*,l}^{-1})}{p_{k^*,l}}  \\
    &=\max_k \left( \log \sum_l \frac{\Tr_{\min(d_{A_l}, d_{B_{k,l}})} (\tau_{k,l}^{-1})}{p_{k,l}} \right).
\end{align}
The expression for the $\cb-$divergences follows similarly.

Finally, the additivity claim follows since the $\max-\cb$ divergence is always additive \cite{wilde_amortized_2020} (alternatively, one can directly check that the stated formula is additive). The additivity claim for any quantum R\'enyi divergence then follows from the fact that $D_{\max}$ upper bounds all quantum R\'enyi divergences (this can be seen for example by noting that $D_{\max}$ is equal to the maximal quantum R\'enyi divergence at $α = \infty$, see e.g.\ \cite[Proposition 4.3]{tomamichel_quantum_2016}); and by using the data-processing inequality with the measurement $\{\Pi_ρ, \iden - \Pi_ρ\}$ (where $\Pi_ρ$ is the orthogonal projection onto the support of $ρ$) to see that $D_{\min}$ lower bounds all quantum R\'enyi divergences.
\end{proof}

\begin{corollary}\label{corollary:PvQ:common-1}
    Let $\cP, \cQ : \B{\Hil}\to \B{\Hil} $ be idempotent channels that share a common full-rank invariant state. Then, for any $\epsilon\in (0,1)$ and $r>D^{\cb}(\cP\Vert \cQ)$:
    \begin{align}
        \lim_{n\to \infty} e_P(\cP,\cQ,n,\epsilon) = \lim_{n\to \infty} e_A(\cP,\cQ,n,\epsilon) &= D^{\cb}(\cP\Vert \cQ) \label{eq:D(P||Q=e)} \\
        e^{\operatorname{sc}}_P(\cP,\cQ, r) = e^{\operatorname{sc}}_A (\cP,\cQ,r) &= r - D^{\cb}(\cP \Vert \cQ) \label{eq:D(P||Q=esc)}.
    \end{align}
    In particular this implies that the strong-converse property holds.
\end{corollary}
\begin{proof}
    If $\im (\cQ^*) \not\subseteq \im (\cP^*)$, then Lemma~\ref{lemma:D(PQ)=inf} shows that $D^{\cb}(\cP \Vert \cQ)= +\infty$. Hence, \[\liminf_{n\to \infty}e_P(\cP, \cQ,n,\epsilon)\geq +\infty.\] Otherwise, if $\im (\cQ^*) \subseteq \im (\cP^*)$, again as before, we prove not just the asymptotic expressions, but slightly stronger $n$-shot bounds. We get the following from \eqref{eq:channel_discrimination_n_shot_exponents_renyi} for all $\alpha\in (0,1)$, $\alpha'>1$, and $\epsilon\in (0,1)$:
    \begin{align}
     D^{\cb}(\cP \Vert \cQ) + \frac{\alpha}{n(\alpha-1)}\log \left( \frac{1}{\epsilon} \right) &\leq e_P(\cP,\cQ,n,\epsilon) \\
     &\leq e_A(\cP, \cQ,n,\epsilon) \\
     &\leq D^{\cb}(\cP \Vert \cQ) + \frac{\alpha'}{n(\alpha'-1)} \log\left( \frac{1}{1-\epsilon} \right),
\end{align}
and Eq.~\eqref{eq:D(P||Q=e)} follows by taking the limit $n\to \infty$. Note that here, we used Theorem~\ref{theorem:PvsQ:common} to infer that $D_{\alpha}^{\cb}(\cP \Vert \cQ)=D^{\cb}(\cP \Vert \cQ) = \sD_{\alpha'}^{\cb}(\cP \Vert \cQ)$ for all $\alpha\in (0,1)$, $\alpha'>1$. Finally, Eq.~\eqref{eq:D(P||Q=esc)} follows from Eq.~\eqref{eq:fawzi-sc} and the additivity result from Theorem~\ref{theorem:PvsQ:common}:
\begin{align}
    e_A^{\mathrm{sc}}(\cP, \cQ, r) = e_P^{\mathrm{sc}}(\cP, \cQ, r) &= \sup_{α > 1} {α - 1 \over α}(r - \widetilde{D}_α^{\cb,\reg}\big(\cP\|\cQ)\big) \\
    &= r - D^{\cb}(\cP \Vert \cQ).
\end{align}
\end{proof}

\begin{corollary}\label{corollary:PvQ:common-2}
Let $\cP, \cQ : \B{\Hil}\to \B{\Hil} $ be idempotent channels that share a common full-rank invariant state. If $\im (\cQ^*)\subseteq \im (\cP^*)$, then
    \begin{align}
        \lim_{n\to \infty} -\frac{1}{n} \log (\perr^{P}(\cP, \cQ, n)) = \lim_{n\to \infty} -\frac{1}{n} \log (\perr^{A}(\cP, \cQ, n)) &= D^{\cb}(\cP \Vert \cQ).
    \end{align}
Similarly, if $\im (\cP^*)\subseteq \im (\cQ^*)$, then
\begin{align}
        \lim_{n\to \infty} -\frac{1}{n} \log (\perr^{P}(\cP, \cQ, n)) = \lim_{n\to \infty} -\frac{1}{n} \log (\perr^{A}(\cP, \cQ, n)) &= D^{\cb}(\cQ \Vert \cP).
    \end{align}
\end{corollary}
\begin{proof}
    It suffices to prove the case $\im (\cQ^*)\subseteq \im (\cP^*)$. From Lemma~\ref{lemma:perr<=}, we can write
    \begin{align}
          \limsup_{n \to \infty} - {1 \over n}\log (\perr^{A}(\cP, \cQ, n)) &\leq D^{\cb}_{\max}(\cP\|\cQ) = D^{\cb}(\cP\Vert \cQ),
     \end{align}     
     where the equality follows from Theorem~\ref{theorem:PvsQ:common}. Similarly, using Lemma~\ref{lemma:perr>=}, we get
     \begin{align}
          \liminf_{n \to \infty} - {1 \over n} \log(\perr^{P}(\cP, \cQ, n)) &\geq \sup_{α \in (0, 1)} (1 - α) D_α^{\cb,\reg}(\cP\|\cQ) \\
          &= D^{\cb}(\cP\Vert \cQ),
    \end{align}
    where the equality again follows from Theorem~\ref{theorem:PvsQ:common}.
\end{proof}
\begin{remark}[Hoeffding Exponents]\label{remark:hoeffding}
    The equality $D_{\min}^{\cb}(\cP\|\cQ) = D^{\cb}(\cP\|\cQ) = D_{\max}^{\cb}(\cP\|\cQ)$ from Theorem~\ref{theorem:PvsQ:common} implies that the optimal asymmetric type-II error exponent can actually be achieved at zero type-I error (since $D_{\min} = D_H^0$). In particular for $r \leq D^{\cb}(\cP\|\cQ)$ this then implies $e^{\mathrm{hoeff}}_{A/P}(\cP, \cQ, r) \coloneqq \overline{e}^{\mathrm{hoeff}}_{A/P}(\cP, \cQ, r) = \underline{e}^{\mathrm{hoeff}}_{A/P}(\cP, \cQ, r) = \infty$ for the Hoeffding exponents
    \begin{align}
        \overline{e}^{\mathrm{hoeff}}_{A/P}(\cP, \cQ, r) &:= \limsup_{n \to \infty} \sup_{\text{A/P strategy}}  \left\{ - {1 \over n} \logα_n \, \big|\, β_n \leq 2^{-rn}\right\}, \\
        \underline{e}^{\mathrm{hoeff}}_{A/P}(\cP, \cQ, r) &:= \liminf_{n \to \infty} \sup_{\text{A/P strategy}}  \left\{ - {1 \over n} \logα_n \, \big| \, β_n \leq 2^{-rn}\right\}\,.
    \end{align}

    Conversely, for $r > D^{\cb}(\cP\|\cQ)$ the results of Corollary~\ref{corollary:PvQ:common-1} imply that
    $e^{\mathrm{hoeff}}_{A/P}(\cP, \cQ, r) = 0$.
\end{remark}

\subsubsection{Fully general case}\label{sec:full-general}

Recall that for idempotent channels $\cP,\cQ:\B{\Hil}\to \B{\Hil}$ with $\cP(\iden_{\Hil}), \cQ(\iden_{\Hil})$ of full rank, the inclusion $\im(\cQ^*)\subseteq \im (\cP^*)$ induces a three-layer decomposition of $\Hil$:
\begin{align}
    \Hil &= \oplus_k \oplus_l A_l \otimes B_{k,l} \otimes C_k \label{eq:H=ABC11} \\
    &= \oplus_k \underbrace{(\oplus_l A_l \otimes B_{k,l})}_{:=D_k} \otimes C_k \label{eq:H=ABC22} \\
    &= \oplus_l A_l \otimes \underbrace{(\oplus_k B_{k,l} \otimes C_k)}_{:=E_l} \label{eq:H=ABC33},
\end{align}
such that
\begin{align}
    \cP &= \bigoplus_k \id_{D_k} \otimes \cR_{\delta_k}  \\
    \cQ &= \bigoplus_l \id_{A_l}\otimes \cR_{\omega_l} \label{eq:Q:l-block}.
\end{align}
When $\cP$ and $\cQ$ do not share a common invariant state, the states $\omega_l$ may carry coherences across $k$-blocks and need not factor across $B_{k,l}\otimes C_k$. The collapse of divergences from the preceding section therefore need not hold. Nevertheless, in this section, we prove that the stabilized sandwiched R\'{e}nyi divergences remain additive for every $\alpha>1$. 

We define the following divergences for each $k,l$:
\begin{align}\label{eq:Dkl-def}
    \sD_{\alpha}(k,l) &:= \sD_{\alpha}(\cP |_{\Pi_{k,l}} \Vert \cQ|_{\Pi_{k,l}}) \\
    &= \sup_{\nu=\Pi_{k,l}\nu \Pi_{k,l}} \sD_{\alpha}( \cP (\nu) \Vert \cQ(\nu)),
\end{align}
where $\cP|_{\Pi_{k,l}}$ and $\cQ|_{\Pi_{k,l}}$ are the restrictions\footnote{For any channel $\Phi:\B{A}\to \B{B}$ and subspace $C\subseteq A$, the restriction $\Phi|_C : \B{C}\to \B{B}$ is defined as $\Phi_C (\cdot) := \Phi (V_C (\cdot) V_C^{\dagger}) $, where $V_C:C\to A$ is the canonical inclusion isometry.} of $\cP$ and $\cQ$ to the subspace $\Pi_{k,l}\Hil = A_{l}\otimes B_{k,l} \otimes C_k$, and the sup is over all states $\nu\in \cD (\Hil)$ satisfying $\nu = \Pi_{k,l}\nu \Pi_{k,l}$. Define canonical inclusion isometries $W_{k|l}:B_{k,l}\otimes C_k \to E_l$ and $W_l:A_l \otimes E_l \to \Hil$ (see Remark~\ref{remark:three-layer}), and the corresponding embedding and compression maps
\begin{equation}
\iota_{k|l}(\cdot):=W_{k|l} (\cdot) W_{k|l}^*, \qquad \cW_l^*(\cdot)=W_l^{\dagger}(\cdot)W_l,
\end{equation}
so that we can restrict the outputs of $\cP|_{\Pi_{k,l}}, \cQ|_{\Pi_{k,l}}$ to the same codomain 
$\B{A_l\otimes E_l}$: 
\begin{align}
\cW_l^* \circ  \cP|_{\Pi_{k,l}} &=
(\id_{A_l}\otimes \iota_{k|l})\circ
(\id_{A_l}\otimes \id_{B_{k,l}}\otimes \cR_{\delta_k}), \label{eq:Pkl}
\\
\cW_l^* \circ \cQ|_{\Pi_{k,l}} &=
\id_{A_l}\otimes \cR_{\omega_l}. \label{eq:Qkl}
\end{align}
Hence,
\begin{equation}\label{eq:Dkl-def2}
\widetilde D_\alpha(k,l) =
\widetilde D_\alpha\!\Bigl(
(\id_{A_l}\otimes \iota_{k|l})\circ
(\id_{A_l}\otimes \id_{B_{k,l}}\otimes \cR_{\delta_k})
\,\Big\|\,
\id_{A_l}\otimes \cR_{\omega_l}
\Bigr).
\end{equation}
Similarly, we define the stabilized divergences
\begin{align}\label{eq:Dklcb-def}
    \sD_{\alpha}^{\cb} (k,l) &:= \sD_{\alpha}^{\cb}(\cP |_{\Pi_{k,l}} \Vert \cQ|_{\Pi_{k,l}}) .
\end{align}

For a completely positive map $\Lambda$ and $\alpha>1$, define the $\alpha\to \alpha$ norm as
\begin{equation}\label{eq:induced-alpha-norm}
    \norm{\Lambda}_{\alpha\to\alpha}
    :=\sup_{Z\ne0}\frac{\norm{\Lambda(Z)}_\alpha}{\norm{Z}_\alpha}
    =\max_{Z\geq0,\,\Tr Z^\alpha=1}\norm{\Lambda(Z)}_\alpha.
\end{equation}
The restriction to positive inputs follows from \cite{Watrous2005pq-norms, Audenaert2009pq-norms},  and the multiplicativity
\begin{equation}\label{eq:induced-alpha-multiplicativity}
    \norm{\Lambda_1\otimes\Lambda_2}_{\alpha\to\alpha}
    =\norm{\Lambda_1}_{\alpha\to\alpha} \norm{\Lambda_2}_{\alpha\to\alpha}
\end{equation}
holds for completely positive maps according to \cite[Theorem 13]{devetak_multiplicativity_2006}.

Using the inclusion isometries $W_{k|l}$ above, define
\begin{equation}\label{eq:Hkl-alpha}
    H_{k,l}^{(\alpha)}
    :=(\iden_{B_{k,l}}\otimes\sqrt{\delta_k})
       W_{k|l}^{\dagger}\omega_l^{\frac{1-\alpha}{\alpha}}W_{k|l}
       (\iden_{B_{k,l}}\otimes\sqrt{\delta_k})
    \in\B{B_{k,l}\otimes C_k}.
\end{equation}
Fix orthonormal bases on the spaces $B_{k,l}$ and define
\begin{align}
    \Lambda_{k,l,\alpha}(Y)
    &:=\Tr_{B_{k,l}}\!\left[
        (Y^{\mathsf T}\otimes\iden_{C_k})H_{k,l}^{(\alpha)}
       \right], \quad Y\in \B{B_{k,l}} \label{eq:Lambda-kl}\\
    \Lambda_{k,\alpha}(Z)
    &:=\sum_l\Lambda_{k,l,\alpha}(Z_{ll}),
    \qquad Z\in\B{\oplus_l B_{k,l}},\label{eq:Lambda-k}
\end{align}
where $Z_{ll}$ is the $l$-th diagonal block of $Z$. Both maps are completely positive, with their (unnormalized) Choi matrices being
\begin{equation}
    H_{k,l}^{(\alpha)} \quad \text{and} \quad \bigoplus_l H_{k,l}^{(\alpha)},
\end{equation}
respectively. All sums below omit zero-dimensional blocks.

\begin{theorem}\label{theorem:PvsQ}
Let $\cP,\cQ:\B{\Hil}\to\B{\Hil}$ be idempotent channels such that $\cP(\iden_{\Hil})$ and $\cQ(\iden_{\Hil})$ have full rank and $\im(\cQ^*)\subseteq\im(\cP^*)$. Then, for every $\alpha>1$,
\begin{equation}\label{eq:PQ-exact-norm}
    \sD_\alpha^{\cb}(\cP\Vert\cQ)
    =\frac{\alpha}{\alpha-1}
      \log\max_k\norm{\Lambda_{k,\alpha}}_{\alpha\to\alpha}.
\end{equation}
Moreover, if $(\cP_j,\cQ_j)$, $j=1,2$, are two pairs satisfying these assumptions, then
\begin{equation}\label{eq:PQ-full-additivity}
    \sD_\alpha^{\cb}(\cP_1\otimes\cP_2\Vert\cQ_1\otimes\cQ_2)
    =\sD_\alpha^{\cb}(\cP_1\Vert\cQ_1)
     +\sD_\alpha^{\cb}(\cP_2\Vert\cQ_2).
\end{equation}
In particular, $\sD_\alpha^{\cb,\reg}(\cP\Vert\cQ)=\sD_\alpha^{\cb}(\cP\Vert\cQ)$.

The restricted divergences from \eqref{eq:Dkl-def} and~\eqref{eq:Dklcb-def} give the upper bounds
\begin{align}
    \sD_\alpha(\cP\Vert\cQ)
    &\leq\max_k\log\sum_l2^{\sD_\alpha(k,l)},\label{eq:Dkl-1}\\
    \sD_\alpha^{\cb}(\cP\Vert\cQ)
    &\leq\max_k\log\sum_l2^{\sD_\alpha^{\cb}(k,l)}.\label{eq:Dkl-2}
\end{align}
Consequently, for every $n\in\N$,
\begin{equation}\label{eq:Dkl-3}
    \sD_\alpha^{\cb}(\cP^{\otimes n}\Vert\cQ^{\otimes n})
    =n\sD_\alpha^{\cb}(\cP\Vert\cQ)
    \leq n\max_k\log\sum_l2^{\sD_\alpha^{\cb}(k,l)}.
\end{equation}
The upper bounds in Eqs.~\eqref{eq:Dkl-1}--\eqref{eq:Dkl-3} are attained when
\begin{equation}\label{eq:omega-block}
    \omega_l=\bigoplus_k p_{k,l}\,\tau_{k,l}\otimes\nu_k,
\end{equation}
for some probability distributions $(p_{k,l})_k$ and full-rank states $\tau_{k,l}\in \cD(B_{k,l}), \nu_k\in \cD (C_k)$.
\end{theorem}

\begin{proof}
We begin with the bound in Eq.~\eqref{eq:Dkl-1}. Note that
    \begin{align}
    \widetilde{D}_{\alpha}(\cP \Vert \cQ) &:= \sup_{\rho\in \State{\Hil}} \widetilde{D}_{\alpha}(\cP(\rho) \Vert \cQ(\rho) )  \\
    &= \sup_{\rho\in \State{\Hil}} \widetilde{D}_{\alpha}(\cP(\rho) \Vert \cQ \circ \cP (\rho) ) \\
    &= \sup_{\rho\in \im(\cP)} \widetilde{D}_{\alpha}( \rho \Vert \cQ (\rho) ) \\
    &= \sup_{\rho\in \im(\cP) \text{ extreme} } \widetilde{D}_{\alpha} \left( \rho \Vert \cQ(\rho) \right),
\end{align}
where $\cQ= \cQ \circ \cP$ follows from $\im (\cQ^*)\subseteq \im (\cP^*)$ and the final equality follows from the joint quasi-convexity of $\widetilde{D}_{\alpha}$ (see e.g. \cite[Chapter 7]{khatriPrinciplesQuantumCommunication2024}). Note that the extreme points of the state space in $\im (\cP)$ are of the form\footnote{Throughout the proof, $V_k, V_{l|k}, W_l, W_{k|l}$ denote the canonical inclusion isometries from Remark~\ref{remark:three-layer}.} 
\begin{equation}
\rho = V_k (\psi \otimes \delta_k) V_k^{\dagger}=0\oplus \ldots\oplus (\psi \otimes \delta_k) \oplus \ldots \oplus 0,    
\end{equation}
supported in $D_k \otimes C_k$ for some $k$ and pure state $\psi\in \B{D_k}$. Let 
\begin{equation}
    \sigma :=\cQ(\rho)=\sum_l W_l \cQ_l (W_l^{\dagger} \rho W_l) W_l^{\dagger}:= \sum_l \sigma_l.
\end{equation}
Recall that each $D_k=\oplus_l A_l \otimes B_{k,l}$. A straightforward calculation shows
\begin{align}
    \sigma_l &= W_l \cQ_l \left(W_l^{\dagger} V_k (\psi \otimes \delta_k) V_k^{\dagger} W_l \right) W_l^{\dagger} \\
    &= p_l W_l \left( \Tr_{B_{k,l}} (\psi_l) \otimes \omega_l \right) W_l^{\dagger} \label{eq:sigma_l},
\end{align} 
where $\psi_l:= V_{l|k}^{\dagger} \psi V_{l|k}/p_l \in \cD (A_l \otimes B_{k,l})$ and $p_l=\Tr (V_{l|k}^{\dagger} \psi V_{l|k})$. Then,
\begin{align}
    \sD_{\alpha}(\rho \Vert \sigma) &= \frac{1}{\alpha-1} \log \Tr \left[ \left(\sigma^{\frac{1-\alpha}{2\alpha}}  \rho\sigma^{\frac{1-\alpha}{2\alpha}} \right)^{\alpha} \right]  \\
    &= \frac{1}{\alpha-1} \log \Tr \left[ \left( \sqrt{\rho} \sigma^{\frac{1-\alpha}{\alpha}} \sqrt{\rho} \right)^{\alpha} \right] \\
    &= \frac{1}{\alpha-1} \log \Tr \left[ \left( V_k(\psi \otimes \sqrt{\delta_k} ) V_k^{\dagger}\sigma^{\frac{1-\alpha}{\alpha}} V_k (\psi \otimes \sqrt{\delta_k} )V_k^{\dagger} \right)^{\alpha} \right] \label{eq:star} \\
    &= \frac{1}{\alpha-1} \log \Tr \bigg[ \bigg( \sqrt{\delta_k}\Tr_{D_k} \left( (\psi \otimes \iden_{C_k}) V_k^{\dagger} \sigma^{\frac{1-\alpha}{\alpha}}V_k \right) \sqrt{\delta_k}   \bigg)^{\alpha}  \bigg ] \label{eq:star2} \\
    &= \frac{1}{\alpha-1} \log \Tr \bigg[ \bigg( \sum_l \sqrt{\delta_k}\Tr_{D_k} \left( (\psi \otimes \iden_{C_k}) V_k^{\dagger} \sigma_l^{\frac{1-\alpha}{\alpha}}V_k \right) \sqrt{\delta_k}   \bigg)^{\alpha}  \bigg ].
\end{align}
To obtain \eqref{eq:star2}, we used the fact that the operator inside the $\alpha$-th power in \eqref{eq:star} is the isometric embedding, via $V_k(\ket{\psi}\otimes\iden_{C_k})$, of the compressed operator on $C_k$ in \eqref{eq:star2}, so their nonzero eigenvalues and the traces of their $\alpha$-th powers agree.

Define positive operators 
\begin{equation}\label{eq:Xkl}
    X_{k,l}:= \sqrt{\delta_k}\Tr_{D_k} \left( (\psi \otimes \iden_{C_k}) V_k^{\dagger} \sigma_l^{\frac{1-\alpha}{\alpha}}V_k \right) \sqrt{\delta_k} \in \B{C_k},
\end{equation}
so that the expression becomes
\begin{align}     
    \sD_{\alpha}(\rho \Vert \sigma) &= \frac{1}{\alpha-1} \log \Tr \left[ \bigg(\sum_l X_{k,l} \bigg)^{\alpha} \right] \label{eq:Xkl-sum} \\
    &\leq \frac{1}{\alpha-1} \log \left[ \sum_l (\Tr X_{k,l}^{\alpha})^{1/\alpha} \right]^{\alpha}, \label{eq:schatten-convex}
\end{align}
where the inequality follows from the convexity of the Schatten norm $\norm{\cdot}_{\alpha}$. Recall that $E_l=\oplus_k B_{k,l}\otimes C_k$. Then, by using the form of $\sigma_l$ from Eq.~\eqref{eq:sigma_l}, we see that 
\begin{equation}\label{eq:Vk σl}
    V_k^{\dagger}\sigma_l V_k = p_l\,(V_{l|k}\otimes \iden_{C_k})
\Bigl(
\Tr_{B_{k,l}}(\psi_l)\otimes W_{k|l}^\dagger \omega_l W_{k|l}
\Bigr)
(V_{l|k}^\dagger\otimes \iden_{C_k})
\end{equation}
is only supported on a single $k,l$ block $A_l\otimes B_{k,l}\otimes C_k$ inside $D_k \otimes C_k$. Hence, by retracing our steps, working exclusively with operators in $\B{A_l\otimes B_{k,l}\otimes C_k}$, we can check that
\begin{align}
\Tr X_{k,l}^{\alpha} = p_l \Tr \bigg( (\psi_l \otimes \sqrt{\delta_k}) \big((\Tr_{B_{k,l}} \psi_l)^{\frac{1-\alpha}{\alpha}} \otimes W_{k|l}^{\dagger}\omega_l^{\frac{1-\alpha}{\alpha}} W_{k|l} \big) (\psi_l \otimes \sqrt{\delta_k}) \bigg)^{\alpha}  
\leq p_l 2^{(\alpha-1)\sD_{\alpha}(k,l)}, \label{eq:Dkl-bound}
\end{align}
where the inequality follows from choosing a particular input $\psi_l \otimes \delta_k \in \B{A_l \otimes B_{k,l} \otimes C_k}$ in \eqref{eq:Dkl-def}, \eqref{eq:Dkl-def2}. Thus, for the input $\rho = V_k (\psi \otimes \delta_k) V_k^{\dagger}$ supported in a single $k-$block, 
\begin{align}
    \sD_{\alpha}(\rho \Vert \cQ(\rho))
    &\leq \frac{\alpha}{\alpha-1} \log \left[ \sum_l p_l^{1/\alpha} 2^{\frac{\alpha-1}{\alpha}\sD_{\alpha} (k,l)} \right] \\
    &\leq \log \left[ \sum_l 2^{\sD_{\alpha} (k,l)} \right],
\end{align}
where the final inequality follows from the fact that, for any set of non-negative numbers $\{c_l \}_l$ with $\sum_l c_l>0$ and $\alpha>1$,
\begin{equation}\label{eq:max-p_l}
    \max_{p_l\geq 0, \sum_l p_l=1} \frac{\alpha}{\alpha-1} \log \sum_l p_l^{\frac{1}{\alpha}} c_l^{\frac{\alpha-1}{\alpha}} = \log \sum_l c_l,
\end{equation}
see Appendix~\ref{appen:aux}. Finally, Eq.~\eqref{eq:Dkl-1} follows by taking a maximum over the $k-$blocks:
\begin{align}
    \widetilde{D}_{\alpha}(\cP \Vert \cQ) &= \sup_{\rho\in \im(\cP) \text{ extreme} } \widetilde{D}_{\alpha} \left( \rho \Vert \cQ(\rho) \right) \\
    &\leq \max_k \log \left[ \sum_l 2^{\sD_{\alpha} (k,l)} \right]. \label{eq:max-kl} 
\end{align}

For the achievability claim, note that in Eq.~\eqref{eq:schatten-convex}, the Schatten norm convexity bound collapses into an equality when the hypothesis of the theorem holds: if $\omega_l = \oplus_k p_{k,l} \, \tau_{k,l} \otimes \nu_k$, then for a fixed $k$, all $X_{k,l}$ operators are proportional to the same positive operator 
\begin{equation}
   X_{k,l}\propto \sqrt{\delta_k} \nu_k^{\frac{1-\alpha}{\alpha}} \sqrt{\delta_k}
\end{equation}
independent of $l$. Therefore, we can select the block $k^*$ that achieves the maximum in Eq.~\eqref{eq:max-kl} and let $\rho^*=0\oplus \ldots\oplus (\psi^* \otimes \delta_{k^*}) \oplus \ldots \oplus 0$ be supported on $D_{k^*} \otimes C_{k^*}$, with $\psi^*\in \cD(D_{k^*})$ chosen in such a way that $\psi^*_l=V^{\dagger}_{l|k^*} \psi^* V_{l|k^*}/p_l^*$ achieves\footnote{To see that this bound is always attained, see e.g.\ \cite{mosonyi_continuity_2024}} the bound in Eq.~\eqref{eq:Dkl-bound}, and $p^*_l = \Tr (V^{\dagger}_{l|k^*} \psi^* V_{l|k^*})$ achieves the bound in Eq.~\eqref{eq:max-p_l}. In the proof's notation, 
\begin{align}
    \sD_{\alpha} ( \rho^* \Vert \cQ(\rho^*) ) &= \frac{1}{\alpha-1} \log \Tr \left[ \bigg(\sum_l X_{k^*,l} \bigg)^{\alpha} \right] \\
    &= \frac{1}{\alpha-1} \log \left[ \sum_l (\Tr X_{k^*,l}^{\alpha})^{1/\alpha} \right]^{\alpha} \\
    &= \frac{\alpha}{\alpha-1} \log \left[ \sum_l (p_l^*)^{1/\alpha} 2^{\frac{\alpha-1}{\alpha}\sD_{\alpha} (k^*,l)} \right] \\
    &= \log \left[ \sum_l 2^{\sD_{\alpha} (k^*,l)} \right] \\
    &= \max_k \log \left[ \sum_l 2^{\sD_{\alpha} (k,l)} \right].
\end{align}
The bound for the $\cb-$divergence (Eq.~\eqref{eq:Dkl-2}) follows by doing exactly the same calculation with an extra reference system attached (Eq.~\eqref{eq:Dcb-def}). 

It remains to establish the exact norm formula and additivity. 
For this, we first show that the $X_{k,l}$ operators can be written using the CP maps $\Lambda_{k,l,\alpha}$ from \eqref{eq:Lambda-kl}:
\begin{equation}\label{eq:Xkl-Lambda}
    X_{k,l}
    =p_l^{1/\alpha}\Tr_{B_{k,l}}\!\left[
       (\tau_l^{1/\alpha}\otimes\iden_{C_k})H_{k,l}^{(\alpha)}
       \right]
    =\Lambda_{k,l,\alpha}(Z_l).
\end{equation}
where $\tau_l:=\Tr_{A_l}\psi_l$ and
$Z_l:=p_l^{1/\alpha}(\tau_l^{\mathsf T})^{1/\alpha}$. Indeed, using \eqref{eq:Vk σl} in \eqref{eq:Xkl}, we can write
\begin{equation}
    X_{k,l}=\sqrt{\delta_k}\Tr_{D_k} \left( (\psi \otimes \iden_{C_k}) p_l^β\,(V_{l|k}\otimes \iden_{C_k})
\Bigl(
[\Tr_{B_{k,l}}(\psi_l)]^β\otimes W_{k|l}^\dagger \omega_l^β W_{k|l}
\Bigr)
(V_{l|k}^\dagger\otimes \iden_{C_k})\right) \sqrt{\delta_k}\,,
\end{equation}
where $β = {1 - α  \over α}$. Now, using again $V_{l|k}^{\dagger} \psi V_{l|k} = p_l \psi_l$ we get
\begin{equation}
    X_{k,l} = p_l^{1 + β} \sqrt{δ_k} \Tr_{D_k}\left[(\psi_l \otimes \iden_{C_k})\Bigl([\Tr_{B_{k,l}}(\psi_l)]^β\otimes W_{k|l}^\dagger \omega_l^β W_{k|l}
\Bigr)\right] \sqrt{δ_k}
\end{equation}
where the expression inside the trace is only supported on $A_l \otimes B_{k,l} \otimes C_k$ so the trace $\Tr_{D_k}$ is equivalent to $\Tr_{A_l, B_{k,l}}$.
Let $\ket{\psi_{l}} = \sum_{i} \sqrt{λ_i} \ket{a_i}_{A_l}\ket{b_i}_{B_{k,l}}$ be a Schmidt-decomposition of $\psi_{l}$ (which is pure). Then $[\Tr_{B_{k,l}}(\psi_l)]^β = \sum_{i} λ_i^β \ket{a_i}\bra{a_i}$, and thus
\begin{equation}
    \Tr_{A_l}[\psi_l ([\Tr_{B_{k,l}}(\psi_l)]^β \otimes \iden_{B_{k,l}})] = \sum_{i} λ_i^{1 + β} \ket{b_i}\bra{b_i} = [\Tr_{A_l} (\psi_l)]^{1 \over α} = \tau_l^{1 \over α}\,.
\end{equation}
Using that $\sqrt{δ_k}$ lives only on $C_k$ and can be pulled inside the trace, as well as the definition of $H_{k,l}^{(α)}$ \eqref{eq:Hkl-alpha} then gives the desired expression for $X_{k,l}$.

Now, we can repeat all the preceding calculations with a reference $R\cong\Hil$. Then, $\psi$ is pure on $R\otimes D_k$, $\psi_l$ is pure on $R\otimes A_l\otimes B_{k,l}$, and $\tau_l=\Tr_{RA_l}\psi_l$. The isometries carry an identity on $R$, and the partial trace in \eqref{eq:Xkl} is over $RD_k$. With these changes, \eqref{eq:Xkl-Lambda} still holds, and the equality in \eqref{eq:Xkl-sum} gives
\begin{equation}\label{eq:Xkl-exact-cb}
    2^{\frac{\alpha-1}{\alpha}\sD_\alpha(\rho\Vert(\id_R\otimes\cQ)(\rho))}
    =\norm{\sum_l X_{k,l}}_\alpha
    =\norm{\sum_l\Lambda_{k,l,\alpha}(Z_l)}_\alpha,
\end{equation}
where
\begin{equation}\label{eq:Zl-constraint}
    Z_l=p_l^{1/\alpha}(\tau_l^{\mathsf T})^{1/\alpha}\geq0,
    \qquad\sum_l\norm{Z_l}_\alpha^\alpha=\sum_l p_l=1.
\end{equation}
Conversely, every family $Z_l\geq0$ with $\sum_l\Tr Z_l^\alpha=1$ can be realized in this way. To see this, given $(Z_{l})_l$, set
\begin{equation}\label{eq:Zl-realization}
    p_l=\Tr Z_l^\alpha,
    \qquad\tau_l=\frac{(Z_l^\alpha)^{\mathsf T}}{p_l}
    \quad(p_l>0),
\end{equation}
purify each $\tau_l$ on $RA_lB_{k,l}$ as $\ket{\psi_l}_{RA_l B_{k,l}}$, and choose
\begin{equation}
    \ket{\psi}_{RD_k}
    =\sum_l\sqrt{p_l}\,(\iden_R\otimes V_{l|k})
       \ket{\psi_l}_{RA_lB_{k,l}}.
\end{equation}
The reference is large enough for every such purification, and the $l$-blocks are orthogonal, so this gives a normalized vector with the required marginals. Without the reference, one would instead have the restriction $\operatorname{rank}\tau_l\leq d_{A_l}$.

Let $\mathcal E_k$ be the pinching onto the diagonal blocks of $\B{\oplus_l B_{k,l}}$. Since $\Lambda_{k,\alpha}\circ\mathcal E_k=\Lambda_{k,\alpha}$, where $\Lambda_{k,\alpha}$ is defined in \eqref{eq:Lambda-k}, and $\norm{\mathcal E_k(Z)}_\alpha\leq\norm{Z}_\alpha$, pinching a nonzero positive input cannot decrease the norm ratio. The positive-input characterization in \eqref{eq:induced-alpha-norm} thus shows
\begin{align}
    2^{\frac{\alpha-1}{\alpha}\sD_\alpha^{\cb}(\cP\Vert\cQ)}
    &=\max_k\max_{Z_l\geq0,\,\sum_l\Tr Z_l^\alpha=1}
         \norm{\sum_l\Lambda_{k,l,\alpha}(Z_l)}_\alpha\\
    &=\max_k\norm{\Lambda_{k,\alpha}}_{\alpha\to\alpha}.
\end{align}
Taking logarithms proves Eq.~\eqref{eq:PQ-exact-norm}.

For two channel pairs, product bases and canonical identifications of tensor factors give
\begin{align}
    H_{(k_1,k_2),(l_1,l_2)}^{(12,\alpha)}
    &=H_{k_1,l_1}^{(1,\alpha)}\otimes H_{k_2,l_2}^{(2,\alpha)},\notag\\
    \Lambda_{(k_1,k_2),\alpha}^{(12)}
    &=\Lambda_{k_1,\alpha}^{(1)}\otimes\Lambda_{k_2,\alpha}^{(2)}.
       \label{eq:Lambda-tensor}
\end{align}
Thus, the multiplicativity of $\alpha\to \alpha$ norm from \eqref{eq:induced-alpha-multiplicativity} implies
\begin{equation}
    \max_{k_1,k_2}\norm{\Lambda_{(k_1,k_2),\alpha}^{(12)}}_{\alpha\to\alpha}
    =\left[\max_{k_1}\norm{\Lambda_{k_1,\alpha}^{(1)}}_{\alpha\to\alpha} \right]
     \left[\max_{k_2}\norm{\Lambda_{k_2,\alpha}^{(2)}}_{\alpha\to\alpha}\right].
\end{equation}
which establishes the required additivity \eqref{eq:PQ-full-additivity} via \eqref{eq:PQ-exact-norm}. 
\end{proof}

The scalar bounds in \eqref{eq:Dkl-1}, \eqref{eq:Dkl-2} can be strict because the Schatten norm convexity step in \eqref{eq:schatten-convex} discards the relative operator structure of $X_{k,l}$, see Appendix~\ref{appen:example} for an example. Exact additivity follows from the norm of their full sum and does not require equality in this step.
The block-product condition in Eq.~\eqref{eq:omega-block} is sufficient for equality in the scalar bounds, but need not be necessary, since proportionality of the optimizing $X_{k,l}$ can occur more generally.

\begin{remark}
When $\cP$ and $\cQ$ share a common full-rank invariant state, Eq.~\eqref{eq:omega-block-prod} implies Eq.~\eqref{eq:omega-block} with $\nu_k=\delta_k$. The scalar bounds then coincide with the explicit expressions in Theorem~\ref{theorem:PvsQ:common}. Thus the common-state assumption gives the additional collapse to those expressions, while the cb-divergences remain additive without it.
\end{remark}

\begin{corollary}\label{corollary:PvsQ}
Let $\cP,\cQ:\B{\Hil}\to\B{\Hil}$ be idempotent channels such that $\cP(\iden_{\Hil})$ and $\cQ(\iden_{\Hil})$ have full rank and $\im(\cQ^*)\subseteq\im(\cP^*)$. Then
\begin{equation}\label{eq:PQ-error}
    D^{\cb,\reg}(\cP\Vert\cQ)=D^{\cb}(\cP\Vert\cQ),
\end{equation}
and, for every $\epsilon\in(0,1)$,
\begin{equation}\label{eq:PQ-error-3}
    \lim_{n\to\infty}e_P(\cP,\cQ,n,\epsilon)
    =\lim_{n\to\infty}e_A(\cP,\cQ,n,\epsilon)
    =D^{\cb}(\cP\Vert\cQ).
\end{equation}
For every $r>D^{\cb}(\cP\Vert\cQ)$, the strong-converse exponents are
\begin{align}
    e_P^{\mathrm{sc}}(\cP,\cQ,r)
    =e_A^{\mathrm{sc}}(\cP,\cQ,r)
    &=\sup_{\alpha>1}\frac{\alpha-1}{\alpha}
       \left(r-\sD_\alpha^{\cb}(\cP\Vert\cQ)\right)\notag\\
    &=\sup_{\alpha>1}\left\{
       \frac{\alpha-1}{\alpha}r
       -\log\max_k\norm{\Lambda_{k,\alpha}}_{\alpha\to\alpha}\right\}.
       \label{eq:PQ-error2}
\end{align}
In particular, these channel pairs satisfy the strong-converse property.
\end{corollary}

\begin{proof}
For any fixed channel pair, the stabilized sandwiched R\'enyi divergence converges to the stabilized Umegaki divergence as $\alpha\downarrow1$ \cite[Lemma 33]{Ding2023forward}. Apply this result at each fixed $n$ to the exact additivity in Theorem~\ref{theorem:PvsQ} to obtain Eq.~\eqref{eq:PQ-error}:
\begin{align*}
    D^{\cb}(\cP^{\otimes n}\Vert\cQ^{\otimes n})
    &=\lim_{\alpha\downarrow1}
       \sD_\alpha^{\cb}(\cP^{\otimes n}\Vert\cQ^{\otimes n})\\
    &=n\lim_{\alpha\downarrow1}\sD_\alpha^{\cb}(\cP\Vert\cQ)
      =nD^{\cb}(\cP\Vert\cQ).
\end{align*}

Next, Eq.~\eqref{eq:channel_discrimination_n_shot_exponents_renyi}
and the additivity established in Theorem~\ref{theorem:PvsQ} give,
for every $n\in\mathbb{N}$, $\alpha\in(0,1)$, $\alpha'>1$,
and $\epsilon\in(0,1)$:
    \begin{align}
     D^{\cb}_{\alpha}(\cP \Vert \cQ) + \frac{\alpha}{n(\alpha-1)}\log \left( \frac{1}{\epsilon} \right) &\leq e_P(\cP,\cQ,n,\epsilon) \\
     &\leq e_A(\cP, \cQ,n,\epsilon) \\
     &\leq \widetilde{D}_{\alpha'}^{\cb}(\cP \Vert \cQ) + \frac{\alpha'}{n(\alpha'-1)} \log\left( \frac{1}{1-\epsilon} \right).
\end{align}
Moreover, \cite[Lemma 33]{Ding2023forward} again shows $\lim_{\alpha\uparrow1}D_\alpha^{\cb}(\cP\Vert\cQ)
    =D^{\cb}(\cP\Vert\cQ)
    =\lim_{\alpha'\downarrow1}
       \widetilde{D}_{\alpha'}^{\cb}(\cP\Vert\cQ).$
Taking $n\to\infty$ first, followed by $\alpha\uparrow1$
and $\alpha'\downarrow1$, therefore yields
\begin{equation}
    D^{\cb}(\cP\Vert\cQ)
    \leq\liminf_{n\to\infty}e_P(\cP,\cQ,n,\epsilon)
    \leq\limsup_{n\to\infty}e_A(\cP,\cQ,n,\epsilon)
    \leq D^{\cb}(\cP\Vert\cQ).
\end{equation}

Finally, substitute $\sD_\alpha^{\cb,\reg}=\sD_\alpha^{\cb}$ into \eqref{eq:fawzi-sc} and use \eqref{eq:PQ-exact-norm} to obtain the last claim.
\end{proof}

\subsection{Pimsner-Popa Indices}
For a unital $*-$subalgebra $\scA = \oplus_l \B{A_l} \otimes \iden_{B_l} \subseteq \B{\Hil}$ with faithful conditional expectation (not necessarily trace-preserving) $E_{\scA}^*$ such that $E_{\scA}=\oplus_l \id_{A_l}\otimes \cR_{\omega_l}$, where $\omega_l\in\cD(B_l)$ are
full-rank states, the \emph{Pimsner-Popa index} is defined as \cite{Pimsner1986index, Gao2020subalgebra-entropy}:
\begin{align}
    C(E_{\scA}) &:= \inf \{ c>0 \, | \, \forall \rho\in \cD (\Hil) : \rho \leq cE_{\scA}(\rho) \}, \\
    C_{\cb}(E_{\scA}) &:= \sup_{n\in \N} C(E_{\scA} \otimes \id_{\M{n}}).
\end{align}
A straightforward calculation shows that 
\begin{align}
    \log C(E_{\scA}) &= D_{\max}(\id \Vert E_{\scA}) \\
    \log C_{\cb}(E_{\scA}) &= D^{\cb}_{\max}(\id \Vert E_{\scA}) .  
\end{align}
When $E_{\scA,\Tr}$ is the unique trace-preserving conditional expectation onto $\scA$, i.e., $\omega_l=\iden_{B_l}/d_{B_l}$ for all $l$, explicit formulae for these indices were derived in \cite[Theorem 6.1]{Pimsner1986index}:
\begin{equation}
    C(E_{\scA, \Tr}) = \sum_l \min(d_{A_l},d_{B_l}) d_{B_l} \quad\text{and}\quad C_{\cb}(E_{\scA, \Tr}) = \sum_l d^2_{B_l},
\end{equation}
which arise as a special case of Theorem~\ref{theorem:IDvQ}. More generally, for the adjoint channel $E_{\scA}=\oplus_l \id_{A_l}\otimes \cR_{\omega_l}$ of an arbitrary faithful conditional expectation $E_{\scA}^*$, Theorem~\ref{theorem:IDvQ} provides explicit expressions for the indices:
\begin{equation}
    C(E_{\scA}) = \sum_l \Tr_{\min(d_{A_l},d_{B_l})} (\omega_l^{-1}) \quad\text{and}\quad C_{\cb}(E_{\scA}) = \sum_l \Tr (\omega_l^{-1}).
\end{equation}

Similarly, for two unital $*-$subalgebras $\scA \subseteq \scB \subseteq \B{\Hil}$ with unique trace-preserving conditional expectations $E_{\scA,\Tr}$ and $E_{\scB, \Tr}$, the \emph{Pimsper-Popa index} is defined as \cite{Pimsner1986index}
\begin{align}
     C(E_{\scA,\Tr}, E_{\scB,\Tr}) &:= \inf \{ c>0 \, | \, \forall \rho\in \cD (\Hil) : E_{\scB,\Tr}(\rho) \leq cE_{\scA, \Tr}(\rho) \}, \label{eq:pimsner-popa} \\
    C_{\cb}(E_{\scA,\Tr}, E_{\scB,\Tr}) &:= \sup_{n\in \N} C(E_{\scA,\Tr} \otimes \id_{\M{n}}, E_{\scB,\Tr} \otimes \id_{\M{n}}). \label{eq:pimsner-popa2}
\end{align}
As before, it is easy to see that 
\begin{align}
    \log C(E_{\scA,\Tr}, E_{\scB,\Tr}) &= D_{\max}(E_{\scB,\Tr} \Vert E_{\scA, \Tr}) \\
    \log C_{\cb}(E_{\scA,\Tr}, E_{\scB,\Tr}) &= D^{\cb}_{\max}(E_{\scB,\Tr} \Vert E_{\scA, \Tr}).  
\end{align}

    Clearly, both $E_{\scA,\Tr}$ and $E_{\scB,\Tr}$ share a common full-rank invariant state, namely the maximally mixed state in $\cD(\Hil)$. Hence, according to Theorem~\ref{theorem:PvsQ:common} (see also Appendix~\ref{appen:three-layer}), we obtain a three-layer decomposition of $\Hil$:
    \begin{align}
    \Hil &= \oplus_k \oplus_l A_l \otimes B_{k,l} \otimes C_k \\
    &= \oplus_k \underbrace{(\oplus_l A_l \otimes B_{k,l})}_{:=D_k} \otimes C_k  \\
    &= \oplus_l A_l \otimes \underbrace{(\oplus_k B_{k,l} \otimes C_k)}_{:=E_l}, 
    \end{align}  
    such that
    \begin{align}
        \scB &= \bigoplus_k \B{D_k} \otimes \iden_{C_k}, \\
        \scA &= \bigoplus_l \B{A_l} \otimes \iden_{E_l}.
    \end{align}
    Moreover, $E_{\scB,\Tr}=\oplus_k \id_{D_k} \otimes \cR_{\delta_k}$ and $E_{\scA, \Tr} = \oplus_l \id_{A_l} \otimes \cR_{\omega_l}$ with 
    \begin{align}
    \delta_{k}&= \frac{\iden_{C_k}}{d_{C_k}}, \quad \text{and} \\
    \omega_l &= \frac{\iden_{E_l}}{d_{E_l}} =\oplus_k \frac{d_{B_{k,l}} d_{C_k}}{\sum_{k'} d_{B_{k',l}} d_{C_k'}} \frac{\iden_{B_{k,l}}}{d_{B_{k,l}}} \otimes \frac{\iden_{C_k}}{d_{C_k}}.
    \end{align}
Theorem~\ref{theorem:PvsQ:common} recovers the formula derived in \cite[Theorem 6.1]{Pimsner1986index} for the indices
\begin{align}
    \log C(E_{\scA, \Tr}\Vert E_{\scB\Tr}) 
    &= \log \max_k  \left(  \sum_l \frac{\min(d_{A_l}, d_{B_{k,l}}) \sum_{k'} d_{B_{k',l}} d_{C_{k'}} }{d_{C_k}} \right) \\
     \log C_{\cb}(E_{\scA, \Tr}\Vert E_{\scB,\Tr}) 
    &= \log \max_k  \left(  \sum_l \frac{d_{B_{k,l}} \sum_{k'} d_{B_{k',l}} d_{C_{k'}} }{d_{C_k}} \right). 
\end{align}
It also generalizes the result in \cite[Theorem 3.1]{Gao2020subalgebra-entropy}, showing that 
\begin{equation}
    \log C(E_{\scA, \Tr}\Vert E_{\scB,\Tr}) = D_{\max}(E_{\scB, \Tr}\Vert E_{\scA,\Tr}) = D(E_{\scB, \Tr}\Vert E_{\scA,\Tr}) = D_{\min}(E_{\scB, \Tr}\Vert E_{\scA,\Tr}).
\end{equation}
In addition, Theorem~\ref{theorem:PvsQ:common} obtains explicit expressions for more general subalgebra inclusion indices that can be defined via arbitrary faithful conditional expectations (not necessarily trace-preserving, with the adjoint channels sharing a common full-rank invariant state) in exactly the same way as in Eq.~\eqref{eq:pimsner-popa}-\eqref{eq:pimsner-popa2}:
\begin{align}
    C(E_{\scA}, E_{\scB}) &:= \inf \{ c>0 \, | \, \forall \rho\in \cD (\Hil) : E_{\scB}(\rho) \leq cE_{\scA}(\rho) \} \\
    C_{\cb}(E_{\scA}, E_{\scB}) &:= \sup_{n\in \N} C(E_{\scA} \otimes \id_{\M{n}}, E_{\scB} \otimes \id_{\M{n}}).
\end{align}

\section{Application}\label{sec:application_gns}
In this section, we apply our main results to GNS-symmetric Markovian dynamics. We prove that discrimination rates for large iterates of a GNS-symmetric channel converge exponentially fast to those of the corresponding idempotent peripheral projections. We begin with some preliminaries on the asymptotic properties of quantum channels. 

Let $\Phi:\B{\Hil}\to \B{\Hil}$ be a quantum channel. Then, $\Phi$ admits a Jordan decomposition 
\begin{equation}
    \Phi = \sum_{i}( \lambda_i \mathcal{P}_i + \mathcal{N}_i) \quad \text{with} \quad \mathcal{N}_i \mathcal{P}_i = \mathcal{P}_i \mathcal{N}_i = \mathcal{N}_i \,\,\, \text{and} \,\,\, \mathcal{P}_i \mathcal{P}_j = \delta_{ij}\mathcal{P}_i,
\end{equation}
where the sum runs over the distinct eigenvalues $\lambda_i$ of $\Phi$, $\mathcal{P}_i$ are projectors whose rank equals the algebraic multiplicity of $\lambda_i$, and $\mathcal{N}_i$ denote the corresponding nilpotent operators \cite[Chapter 6]{wolf_quantum_2012}. All the eigenvalues $\lambda_i$ of $\Phi$ satisfy $\abs{\lambda_i}\leq 1$ and they are either real or come in complex conjugate pairs. Since $\Phi$ always admits a fixed point, $\lambda=1$ is always an eigenvalue of $\Phi$. Moreover, all $\lambda_i$ with $\abs{\lambda_i}=1$ have equal algebraic and geometric multiplicities, so that $\mathcal{N}_i=0$ for all such eigenvalues. As $l\to \infty$, the image of 
\begin{equation}
    \Phi^l := \underbrace{\Phi \circ \Phi \circ \ldots \circ \Phi}_{l\,  \text{times}}
\end{equation}
converges to the \emph{peripheral} space $\mathscr{X} (\Phi):= \text{span}\{X\in \B{\Hil} : \exists \,\theta\in \mathbb{R} \text{ s.t.\ } \Phi(X)= e^{i\theta} X\}$. 
\begin{definition}\label{def:peripheral-proj}
   Let $\Phi:\B{\Hil}\to \B{\Hil}$ be a quantum channel. The asymptotic part of $\Phi$ and the projector onto the peripheral space $\mathscr{X} (\Phi)$, are respectively defined as follows: 
   \begin{equation}\label{eq:phiinf-proj}
\Phi_{\infty}:= \sum_{i:\, |\lambda_i|=1}\lambda_i \mathcal{P}_i \quad \text{and} \quad  \mathcal{P}_{\Phi} = \sum_{i: \, |\lambda_i|=1} \mathcal{P}_i .
\end{equation}
Note that $\Phi_{\infty}=\Phi\circ \cP_{\Phi}=\cP_{\Phi}\circ \Phi$ and $\im (\cP_{\Phi})=\mathscr{X}(\Phi)$. We denote the largest magnitude of a non-peripheral eigenvalue of $\Phi$ by $\mu_{\Phi}:=\operatorname{spr}(\Phi-\Phi_{\infty})$, so that $1-\mu_{\Phi}$ is the spectral gap\footnote{Here, $\operatorname{spr}(T)$ denotes the spectral radius of a linear map $T$.}.
\end{definition}

Given arbitrary channels $\Phi, \Psi : \B{\Hil}\to \B{\Hil}$, the divergence sequence $D(\Phi^l \Vert \Psi^l)$ can dramatically oscillate. For example, let $\Phi=\id_{\C^2}$ be the qubit identity channel and $\Psi (\cdot)=X(\cdot)X$ be the Pauli-$X$ channel. Then, for all $k\in \mathbb{N}$:
\begin{itemize}
    \item $D(\Phi^{2k}\Vert \Psi^{2k}) = 0$, and
    \item $D(\Phi^{2k+1} \Vert \Psi^{2k+1})= +\infty$.
\end{itemize}

\subsection{GNS-symmetric channels}

Consider now a channel $\Phi:\B{\Hil}\to \B{\Hil}$ that satisfies GNS-detailed balance with respect to a full-rank invariant state $\delta=\Phi(\delta)$:
\begin{align}\label{eq:GNS}
  \forall X,Y\in \B{\Hil} : \quad 
   \Tr (\Phi^{*}(X)Y\delta) &= \Tr (X \Phi^{*}(Y)\delta).
\end{align}
Then, $\Phi$ satisfies (complete) strong data-processing inequality (CSDPI) with respect to its peripheral part, i.e.\ for all $\rho\in \cD(\Hil \otimes \Hil)$ and $l\in \mathbb{N}$ \cite{Gao2022sobolev, Gao2025sobolev},
\begin{align}
    D \left((\id_{\Hil} \otimes\Phi^l) (\rho) \Vert (\id_{\Hil} \otimes \Phi^l\circ\cP_{\Phi}) (\rho) \right) &\leq s_{\Phi}^l D( \rho \Vert (\id_{\Hil} \otimes \cP_{\Phi}) (\rho)) 
\end{align}
for some constant $s_{\Phi}<1$.
Hence, Pinsker's inequality shows that for all $l\in \mathbb{N}$:
\begin{align}
    \norm{\Phi^l-\Phi^l_{\infty}}_{\diamond} & 
    \leq \sqrt{2 s_{\Phi}^l D^{\cb}(\id_{\Hil} \Vert \cP_{\Phi}) } \xrightarrow[l\to \infty]{} 0 
\end{align}

Moreover, GNS-symmetry forces $\spec (\Phi) \subseteq [-1, 1 ]$, so there is only a possibility of 2-cycles in the periphery. More precisely, one can show \cite{Gao2022sobolev}:
\begin{align}
    \Phi^{2k} \circ \cP_{\Phi} = \cP_{\Phi} \circ \Phi^{2k} = \cP_{\Phi}  \quad\text{and}\quad \norm{\Phi^{2k} - \cP_{\Phi}}_{\diamond} \leq \sqrt{2 s_{\Phi}^{2k} D^{\cb}(\id_{\Hil} \Vert \cP_{\Phi})} \xrightarrow[k\to\infty]{} 0.
\end{align}

\begin{theorem}\label{theorem:GNS}
    Suppose $\Phi,\Psi:\B{\Hil}\to \B{\Hil}$ are GNS-symmetric with respect to a common full-rank invariant state $\tau=\Phi(\tau)=\Psi(\tau)$ and $\im( \cP_{\Psi}^*) \subseteq \im (\cP_{\Phi}^*)$. Then, for 
     \begin{equation}
        2k > \max \left\{ \frac{\log C_\Phi}{\log(1/\mu_\Phi)}, \frac{\log C_\Psi}{\log(1/\mu_\Psi)} \right\},
    \end{equation}
       where $C_\star = 2^{D^{\cb}(\id_\Hil \| \cP_\star )}$ for $\star\in \{\Phi,\Psi \}$, we have 
    \begin{align}
     D^{\cb}(\cP_{\Phi} \Vert \cP_{\Psi} ) + \log \left( \frac{1}{1 + \epsilon_{\Psi}^{2k}} \right) &\leq D^{\cb}_{\min}(\Phi^{2k} \Vert \Psi^{2k}) \\
     &\leq D^{\cb}_{\max}(\Phi^{2k} \Vert \Psi^{2k}) \\
     &\leq  D^{\cb}(\cP_{\Phi} \Vert \cP_{\Psi})+ \log \big(1 + \epsilon_{\Phi}^{2k} \big)  + \log \left( \frac{1}{1 - \epsilon_{\Psi}^{2k}} \right),
    \end{align}
    where $\epsilon_{\star}^{2k}:= \mu^{2k}_{\star} C_\star \xrightarrow[k\to \infty]{} 0$ for $\star\in \{\Phi,\Psi \}$.
\end{theorem}

\begin{proof}
    For the upper bound, we use the chain rule (see \cite[Theorem III.1]{Christandl2017max}):
    \begin{align}
        D^{\cb}_{\max}(\Phi^{2k} \Vert \Psi^{2k}) &\leq D^{\cb}_{\max}(\Phi^{2k} \Vert \cP_{\Phi}) + D^{\cb}_{\max}(\cP_{\Phi} \Vert \cP_{\Psi}) + D^{\cb}_{\max}(\cP_{\Psi} \Vert \Psi^{2k}) \\ 
        &\leq \log \big(1 + \epsilon_{\Phi}^{2k} \big) + \log \left( \frac{1}{1 - \epsilon_{\Psi}^{2k}} \right)  + D^{\cb}(\cP_{\Phi} \Vert \cP_{\Psi}),
    \end{align}
    where the second inequality follows from Theorem~\ref{theorem:PvsQ:common} and the cp-order mixing estimates \cite[see Lemma 5.1]{Bardet2024decay}:
    \begin{align}
       (1 - \epsilon_{\Phi}^{2k})  \cP_{\Phi} \leq_{\cp} \Phi^{2k} &\leq_{\cp} (1 + \epsilon_{\Phi}^{2k})  \cP_{\Phi} \label{eq:cpmixing-Phi} \\
       (1 - \epsilon_{\Psi}^{2k}) \cP_{\Psi} \leq_{\cp} \Psi^{2k} &\leq_{\cp} (1 + \epsilon_{\Psi}^{2k})  \cP_{\Psi}.\label{eq:cpmixing-Psi}
    \end{align}
    For the lower bound, we note that,
    \begin{align}
        D^{\cb}_{\min}(\Phi^{2k} \Vert \Psi^{2k}) &\geq D^{\cb}_{\min}(\Phi^{2k} \Vert \cP_{\Psi} ) + \log \left( \frac{1}{1 + \epsilon_{\Psi}^{2k}} \right) \\ 
        &\geq D^{\cb}_{\min}(\cP_{\Phi} \Vert \cP_{\Psi} )  + \log \left( \frac{1}{1 + \epsilon_{\Psi}^{2k}} \right), \\
        &= D^{\cb}(\cP_{\Phi} \Vert \cP_{\Psi} )  + \log \left( \frac{1}{1 + \epsilon_{\Psi}^{2k}} \right),
    \end{align}
    where the first inequality follows from the mixing estimate Eq.~\eqref{eq:cpmixing-Psi} because $\sigma \leq \sigma'$ implies $D_{\min}(\rho \Vert \sigma)\geq D_{\min}(\rho \Vert \sigma')$, the second inequality follows from monotonicity of the channel divergence under common preprocessing:
     \begin{align}
        D^{\cb}_{\min}(\Phi^{2k} \Vert \cP_{\Psi} )  \geq D^{\cb}_{\min}(\Phi^{2k} \circ \cP_{\Phi} \Vert \cP_{\Psi} \circ \cP_{\Phi} ) = D^{\cb}_{\min}(\cP_{\Phi} \Vert \cP_{\Psi} ),
    \end{align}
    where $\cP_{\Psi}\circ \cP_{\Phi}=\cP_{\Psi}$ follows from the image inclusion $\im(\cP_{\Psi}^*)\subseteq \im (\cP_{\Phi}^*)$, and the final equality $D^{\cb}_{\min}(\cP_{\Phi} \Vert \cP_{\Psi} )=D^{\cb}(\cP_{\Phi} \Vert \cP_{\Psi} )$ follows from Theorem~\ref{theorem:PvsQ:common}.
\end{proof}

\begin{corollary}
     Suppose $\Phi,\Psi:\B{\Hil}\to \B{\Hil}$ are GNS-symmetric with respect to a common full-rank invariant state $\tau=\Phi(\tau)=\Psi(\tau)$ and $\im( \cP_{\Psi}^*) \subseteq \im (\cP_{\Phi}^*)$. Then, for 
      \begin{equation}
        2k > \max \left\{ \frac{\log C_\Phi}{\log(1/\mu_\Phi)}, \frac{\log C_\Psi}{\log(1/\mu_\Psi)} \right\},
    \end{equation}
    where $C_\star = 2^{D^{\cb}(\id_\Hil \| \cP_\star )}$ for $\star\in \{\Phi,\Psi \}$,
    and any $\epsilon\in (0,1)$, the following bounds hold: 
    \begin{align}
     D^{\cb}(\cP_{\Phi} \Vert \cP_{\Psi} ) + \log \left( \frac{1}{1 + \epsilon_{\Psi}^{2k}} \right) &\leq  \liminf_{n\to \infty} e_P (\Phi^{2k},\Psi^{2k},n,\epsilon) \\ 
     &\leq \limsup_{n\to \infty} e_A (\Phi^{2k},\Psi^{2k},n,\epsilon)  \\
    &\leq  D^{\cb}(\cP_{\Phi} \Vert \cP_{\Psi})+ \log \big(1 + \epsilon_{\Phi}^{2k} \big)  + \log \left( \frac{1}{1 - \epsilon_{\Psi}^{2k}} \right),
    \end{align}
    where $\epsilon_{\star}^{2k}:= \mu^{2k}_{\star}C_\star \xrightarrow[k\to \infty]{} 0$ for $\star\in \{\Phi,\Psi \}$. 
    
    Furthermore, if $r > D^{\cb}(\cP_{\Phi} \| \cP_{\Psi}) + \log \left( \frac{1}{1+ \epsilon_{\Psi}^{2k}} \right),$ then 
    \begin{equation}
    e^{\mathrm{sc}}_P(\Phi^{2k}, \Psi^{2k}, r)
    = e^{\mathrm{sc}}_A(\Phi^{2k}, \Psi^{2k}, r) \leq r - D^{\cb}(\cP_{\Phi} \| \cP_{\Psi}) - \log \left( \frac{1}{1 + \epsilon_{\Psi}^{2k}} \right),
    \end{equation}
    and if $r > D^{\cb}(\cP_{\Phi} \Vert \cP_{\Psi})+ \log \big(1 + \epsilon_{\Phi}^{2k} \big)  + \log \left( \frac{1}{1 - \epsilon_{\Psi}^{2k}} \right),$ then  
    \begin{equation}
    e^{\mathrm{sc}}_P(\Phi^{2k}, \Psi^{2k}, r)
    = e^{\mathrm{sc}}_A(\Phi^{2k}, \Psi^{2k}, r) \geq r - D^{\cb}(\cP_{\Phi} \Vert \cP_{\Psi}) - \log \big(1 + \epsilon_{\Phi}^{2k} \big)  - \log \bigg( \frac{1}{1 - \epsilon_{\Psi}^{2k}} \bigg).
    \end{equation}
\end{corollary}
\begin{proof}
    For the first claim, we establish the following slightly stronger $n$-shot bounds. From Eq.~\eqref{eq:channel_discrimination_n_shot_exponents_renyi}, the following holds for all $\epsilon\in (0,1)$ and $n,k\in \mathbb{N}$:
    \begin{align}
    e_P(\Phi^{2k},\Psi^{2k},n,\epsilon) \geq D_{\min}^{\cb}(\Phi^{2k} \| \Psi^{2k}) \geq D^{\cb}(\cP_{\Phi} \Vert \cP_{\Psi} ) &+ \log \left( \frac{1}{1 + \epsilon_{\Psi}^{2k}} \right),
\end{align}
where the last bound follows from Theorem~\ref{theorem:GNS}. Similarly, using Eq.~\eqref{eq:channel_discrimination_n_shot_exponents_renyi} again, 
\begin{align}   
    e_A(\Phi^{2k}, \Psi^{2k},n,\epsilon)
     &\leq D_{\max}^{\cb}(\Phi^{2k} \Vert \Psi^{2k}) + \frac{1}{n} \log\left( \frac{1}{1-\epsilon} \right) \\
    &\leq  D^{\cb}(\cP_{\Phi} \Vert \cP_{\Psi})+ \log \big(1 + \epsilon_{\Phi}^{2k} \big)  + \log \left( \frac{1}{1 - \epsilon_{\Psi}^{2k}} \right) +\frac{1}{n} \log\left( \frac{1}{1-\epsilon} \right),
\end{align}
where the last bound again follows from Theorem~\ref{theorem:GNS}. The desired claim then follows by taking the limit $n\to \infty$.

The bounds on the strong converse exponents follow from Eq.~\eqref{eq:fawzi-sc}:
\begin{align}
    e_P^{\mathrm{sc}}(\Phi^{2k}, \Psi^{2k}, r) &= e_A^{\mathrm{sc}}(\Phi^{2k}, \Psi^{2k}, r) \\ 
    &= \sup_{α > 1} {α - 1 \over α}(r - \widetilde{D}_α^{\cb, \reg}\big(\Phi^{2k} \|\Psi^{2k})\big)  \\
    &\geq r - D^{\cb}(\cP_{\Phi} \Vert \cP_{\Psi}) - \log \big(1 + \epsilon_{\Phi}^{2k} \big)  - \log \left( \frac{1}{1 - \epsilon_{\Psi}^{2k}} \right),
\end{align}
where we used the fact that $\widetilde{D}_α^{\cb, \reg}\big(\Phi^{2k} \|\Psi^{2k}) \leq D^{\cb}_{\max}(\Phi^{2k} \| \Psi^{2k})$, after which Theorem~\ref{theorem:GNS} is used to bound $D^{\cb}_{\max}$ as before. Similarly, 
\begin{align}
    e_P^{\mathrm{sc}}(\Phi^{2k}, \Psi^{2k}, r) &= e_A^{\mathrm{sc}}(\Phi^{2k}, \Psi^{2k}, r) \\ 
    &= \sup_{α > 1} {α - 1 \over α}(r - \widetilde{D}_α^{\cb, \reg}\big(\Phi^{2k} \|\Psi^{2k})\big) \\
    &\leq r - D^{\cb}(\cP_{\Phi} \Vert \cP_{\Psi}) - \log \left( \frac{1}{1 + \epsilon_{\Psi}^{2k}} \right),
\end{align}
where we used the fact that $\widetilde{D}_α^{\cb, \reg}\big(\Phi^{2k} \|\Psi^{2k}) \geq D^{\cb}_{\min}(\Phi^{2k} \| \Psi^{2k})$, after which Theorem~\ref{theorem:GNS} is used to bound $D^{\cb}_{\min}$ as before.
\end{proof}

\begin{corollary}
  Suppose $\Phi,\Psi:\B{\Hil}\to \B{\Hil}$ are GNS-symmetric with respect to a common full-rank invariant state $\tau=\Phi(\tau)=\Psi(\tau)$. Suppose
    \begin{equation}
        2k > \max \left\{ \frac{\log C_\Phi}{\log(1/\mu_\Phi)}, \frac{\log C_\Psi}{\log(1/\mu_\Psi)} \right\},
    \end{equation}
    where $C_\star = 2^{D^{\cb}(\id_\Hil \| \cP_\star )}$ for $\star\in \{\Phi,\Psi \}$. Then, if $\im( \cP_{\Psi}^*) \subseteq \im (\cP_{\Phi}^*)$ the following holds true:
    \begin{align}
     D^{\cb}(\cP_{\Phi} \Vert \cP_{\Psi} ) + \log \left( \frac{1}{1 + \epsilon_{\Psi}^{2k}} \right) 
     &\leq \liminf_{n\to \infty} -\frac{1}{n} \log (\perr^{P}(\Phi^{2k}, \Psi^{2k}, n)) \\ 
        &\leq \limsup_{n\to \infty} -\frac{1}{n} \log (\perr^{A}(\Phi^{2k}, \Psi^{2k}, n)) \\ 
        &\leq D^{\cb}(\cP_{\Phi} \Vert \cP_{\Psi})+ \log \big(1 + \epsilon_{\Phi}^{2k} \big)  + \log \left( \frac{1}{1 - \epsilon_{\Psi}^{2k}} \right),
    \end{align}
    where $\epsilon_{\star}^{2k}:= \mu^{2k}_{\star}C_\star \xrightarrow[k\to \infty]{} 0$ for $\star\in \{\Phi,\Psi \}$. Similarly, if $\im (\cP_{\Phi}^*)\subseteq \im (\cP_{\Psi}^*)$, then
 \begin{align}
     D^{\cb}(\cP_{\Psi} \Vert \cP_{\Phi} ) + \log \left( \frac{1}{1 + \epsilon_{\Phi}^{2k}} \right)  &\leq \liminf_{n\to \infty} -\frac{1}{n} \log (\perr^{P}(\Phi^{2k}, \Psi^{2k}, n)) \\ 
        &\leq \limsup_{n\to \infty} -\frac{1}{n} \log (\perr^{A}(\Phi^{2k}, \Psi^{2k}, n)) \\ 
        &\leq D^{\cb}(\cP_{\Psi} \Vert \cP_{\Phi})+ \log \big(1 + \epsilon_{\Psi}^{2k} \big)  + \log \left( \frac{1}{1 - \epsilon_{\Phi}^{2k}} \right).
    \end{align}
\end{corollary}

\begin{proof}
    It suffices to prove the case $\im( \cP_{\Psi}^*) \subseteq \im (\cP_{\Phi}^*)$. From Lemma~\ref{lemma:perr<=}, we can write
    \begin{align}
          \limsup_{n \to \infty} - {1 \over n}\log (\perr^{A}(\Phi^{2k}, \Psi^{2k}, n)) &\leq D^{\cb}_{\max}(\Phi^{2k} \|\Psi^{2k} ) \\ 
          &\leq D^{\cb}(\cP_{\Phi} \Vert \cP_{\Psi})+ \log \big(1 + \epsilon_{\Phi}^{2k} \big)  + \log \left( \frac{1}{1 - \epsilon_{\Psi}^{2k}} \right),
     \end{align}     
     where the last bound follows from Theorem~\ref{theorem:GNS}. Similarly, using Lemma~\ref{lemma:perr>=}, we get
     \begin{align}
          \liminf_{n \to \infty} - {1 \over n} \log(\perr^{P}(\Phi^{2k}, \Psi^{2k}, n)) &\geq \sup_{α \in (0, 1)} (1 - α) D_α^{\cb, \reg}(\Phi^{2k}\|\Psi^{2k}) \\
          &\geq D^{\cb}_{\min}(\Phi^{2k} \| \Psi^{2k}) \\
          &\geq D^{\cb}(\cP_{\Phi} \Vert \cP_{\Psi} ) + \log \left( \frac{1}{1 + \epsilon_{\Psi}^{2k}} \right) 
    \end{align}
    where the last bound again follows from Theorem~\ref{theorem:GNS}.
\end{proof}

\section*{Acknowledgements}
SS acknowledges support from the Deutsche Forschungsgemeinschaft (DFG, German Research Foundation) via TRR 352 – Project-ID 470903074. BB acknowledges support from EPSRC UK under grant number EP/Y028732/1. 

All results in the first version of this manuscript were obtained \emph{without} any form of AI assistance. In the current version, ChatGPT 6 Astra was used to prove that the stabilized sandwiched R\'enyi divergence $\widetilde{D}^{\cb}_{\alpha}(\cP \| \cQ)$ can be represented exactly via the induced Schatten $\alpha\to\alpha$ norms of auxiliary completely positive maps (see \eqref{eq:PQ-exact-norm}). Since these maps respect tensor products, norm multiplicativity \cite[Theorem 13]{devetak_multiplicativity_2006} yields additivity of channel divergence for every $\alpha>1$ (see \eqref{eq:PQ-full-additivity}) without the common full-rank invariant state assumption. The crucial new step is a rewriting of the $X_{k,l}$ operators from the earlier proof in terms of suitably defined completely positive maps (see \eqref{eq:Xkl-Lambda}). The authors have verified all AI generated material and are responsible for the final presentation and any remaining errors.

\section*{Appendices}
\addappheadtotoc
\begin{subappendices}
\renewcommand{\setthesubsection}{\Alph{subsection}}
\counterwithin{theorem}{subsection}
\renewcommand{\thetheorem}{\thesubsection.\arabic{theorem}}
\subsection{Three-layer decomposition} \label{appen:three-layer}

\begin{lemma} \label{lemma:three-layer}
Let $\mathcal \scA\subseteq \mathcal \scB\subseteq \B{\Hil}$ be finite-dimensional unital $*$-subalgebras. Then, there exist finite-dimensional Hilbert
spaces $A_l$, $B_{k,l}$, $C_k$, and a unitary
\begin{equation}
U:\mathcal H\to \bigoplus_{k,l} A_l\otimes B_{k,l}\otimes C_k
\end{equation}
such that, with $D_k:=\bigoplus_l A_l\otimes B_{k,l}$ and $
E_l:=\bigoplus_k B_{k,l}\otimes C_k,$
one has
\begin{equation}
U \scB U^\dagger
=
\bigoplus_k \B{D_k}\otimes \iden_{C_k},
\qquad
U \scA  U^\dagger
=
\bigoplus_l \B{A_l}\otimes \iden_{E_l}.
\end{equation}
\end{lemma}

\begin{proof}
Since we are in finite dimensions, both $\mathscr{A}$ and $\mathscr{B}$ are $*$-isomorphic to direct sums of matrix algebras \cite[Theorem 11.2]{Takesaki1979algebra}: $\scB \cong \oplus_{k=1}^K \M{d_k}$ and $\scA \cong \oplus_{l=1}^L \M{a_l}$, both canonically represented on $\Hil$. Hence, there exist Hilbert spaces $D_k,C_k$ with $\dim D_k=d_k$ and a
unitary \cite[Theorem 11.9]{Takesaki1979algebra}
\begin{equation}
U^{(1)}:\mathcal H\to \bigoplus_k D_k\otimes C_k
\end{equation}
such that
\begin{equation}
U^{(1)}\scB (U^{(1)})^\dagger
=
\bigoplus_k \B{D_k}\otimes \iden_{C_k}.
\end{equation}
Since $\scA\subseteq \scB$, the algebra $U^{(1)}\scA(U^{(1)})^\dagger$ preserves each summand $D_k\otimes C_k$ and acts trivially on $C_k$, so for each $k$ it
determines a representation of $\scA$ on $D_k$. Applying \cite[Theorem 11.9]{Takesaki1979algebra} again, for
each $k$ there exist Hilbert spaces $A_l,B_{k,l}$ with $\dim A_l = a_l$ and a unitary
\begin{equation}
U_k^{(2)}:D_k\to \bigoplus_l A_l\otimes B_{k,l}
\end{equation}
such that
\begin{equation}
U_k^{(2)}\bigl(U^{(1)}\mathcal \scA(U^{(1)})^\dagger|_{D_k}\bigr)(U_k^{(2)})^\dagger
=
\bigoplus_l \B{A_l}\otimes \iden_{B_{k,l}}.
\end{equation}
Now, set
\begin{equation}
U:=\Bigl(\bigoplus_k (U_k^{(2)}\otimes \iden_{C_k})\Bigr)U^{(1)}.
\end{equation}
Then
\begin{equation}
U\mathcal H
=
\bigoplus_k \Bigl(\bigoplus_l A_l\otimes B_{k,l}\Bigr)\otimes C_k
=
\bigoplus_{k,l} A_l\otimes B_{k,l}\otimes C_k,
\end{equation}
and the asserted forms of $U\mathcal \scA U^\dagger$ and $U\mathcal \scB U^\dagger$ follow
immediately. The equivalent descriptions via $D_k$ and $E_l$ are obtained by regrouping the
same direct sum.
\end{proof}

\begin{remark}\label{remark:three-layer}
For idempotent channels $\cP,\cQ:\B{\Hil}\to \B{\Hil}$ with $\cP(\iden_{\Hil}), \cQ(\iden_{\Hil})$ of full rank, since both $\im(\cQ^*)$ and $ \im (\cP^*)$ are unital $*$-subalgebras in $\B{\Hil}$ \cite{Lindblad1999fixed}, the inclusion $\im(\cQ^*)\subseteq \im (\cP^*)$ induces a three-layer decomposition of $\Hil$ as obtained in Lemma~\ref{lemma:three-layer}. After fixing a unitary $U$ as in the lemma, define
\begin{equation}
\widetilde \cP:=\operatorname{Ad}_U\circ \cP\circ \operatorname{Ad}_{U^\dagger},
\qquad
\widetilde \cQ:=\operatorname{Ad}_U\circ \cQ\circ \operatorname{Ad}_{U^\dagger},
\end{equation}
where $\operatorname{Ad}_U (\cdot)=U (\cdot) U^{\dagger}$. Since all channel divergences considered in this paper are invariant under unitary
conjugation, one may replace $(\cP,\cQ)$ by $(\widetilde \cP,\widetilde \cQ)$ without changing any of
the quantities of interest. Thus, after this identification, we will always write
\begin{equation}
\Hil =\bigoplus_{k,l} A_l\otimes B_{k,l}\otimes C_k
=\bigoplus_k D_k\otimes C_k
=\bigoplus_l A_l\otimes E_l
\end{equation}
as literal equalities, and
\begin{equation}
\cP=\bigoplus_k \id_{D_k}\otimes \cR_{\delta_k},
\qquad
\cQ=\bigoplus_l \id_{A_l}\otimes \cR_{\omega_l}.
\end{equation}
In this representation, all partial traces are ordinary tensor-factor partial traces, and
\begin{alignat}{2}
    V_k : D_k \otimes C_k &\to \Hil \qquad\quad \,\,\,\, W_l : A_l \otimes E_l &&\to \Hil \\
    V_{l|k}: A_l \otimes B_{k,l} &\to D_k \qquad W_{k|l}:B_{k,l}\otimes C_k &&\to E_l
\end{alignat}
are simply the canonical inclusion isometries associated with the direct-sum decompositions.
\end{remark}

\subsection{Example} \label{appen:example}

Recall that for idempotent channels $\cP,\cQ:\B{\Hil}\to \B{\Hil}$ with $\cP(\iden_{\Hil}), \cQ(\iden_{\Hil})$ of full rank, the inclusion $\im(\cQ^*)\subseteq \im (\cP^*)$ induces a three-layer decomposition of $\Hil$:
\begin{align}
    \Hil &= \oplus_{k=1}^K \oplus_{l=1}^L A_l \otimes B_{k,l} \otimes C_k  \\
    &= \oplus_k \underbrace{(\oplus_l A_l \otimes B_{k,l})}_{:=D_k} \otimes C_k  \\
    &= \oplus_l A_l \otimes \underbrace{(\oplus_k B_{k,l} \otimes C_k)}_{:=E_l},
\end{align}
such that
\begin{align}
    \cP &= \bigoplus_k \id_{D_k} \otimes \cR_{\delta_k}  \\
    \cQ &= \bigoplus_l \id_{A_l}\otimes \cR_{\omega_l} .
\end{align}

Now, let $K=1, L=2$ and 
\begin{equation}
A_1=A_2=\mathbb C,\qquad
B_{1,1}=B_{1,2}=\mathbb C^2,\qquad
C_1=\mathbb C^2,\qquad
\delta_1=\frac12\,\iden_{C_1}.
\end{equation}
Thus, 
\begin{align}
    \Hil &= \oplus_{l=1}^2  B_{1,l} \otimes C_1  \\
    &= \underbrace{(B_{1,1} \oplus B_{1,2})}_{:=D_1} \otimes C_1  \\
    &= \underbrace{ (B_{1,1} \otimes C_1)}_{:=E_1} \oplus  \underbrace{ (B_{1,2} \otimes C_1)}_{:=E_2}.
\end{align}
Define states $\omega_l\in \cD(E_l)$ in the basis $\{ \ket{0}_B\ket{0}_C, \ket{0}_B\ket{1}_C, \ket{1}_B\ket{0}_C, \ket{1}_B\ket{1}_C \}$ as follows:
\begin{align}
    \omega_1 := \frac18\,\operatorname{diag}(1,2,3,2), \qquad
    \omega_2 :=\frac18\,\operatorname{diag}(2,1,2,3).
\end{align}

Neither $\omega_1$ nor $\omega_2$ is of product form across the
$B_{1,l}$--$C_1$ cut, since for a diagonal product state
$\operatorname{diag}(a,b,c,d)$ one necessarily has $ad=bc$. Take $\alpha=2$, so that $
\beta=\frac{1-\alpha}{\alpha}=-\frac12$. Then,
\begin{equation}
\omega_1^{-1/2}
=
\operatorname{diag}\!\left(2\sqrt2,\ 2,\ 2\sqrt{\frac23},\ 2\right),
\qquad
\omega_2^{-1/2}
=
\operatorname{diag}\!\left(2,\ 2\sqrt2,\ 2,\ 2\sqrt{\frac23}\right).
\end{equation}

\smallskip
\noindent
\emph{Restricted block divergences.} \\
Since $A_l=\mathbb C$, the restricted channels (Eqs.~\eqref{eq:Pkl}, \eqref{eq:Qkl})
are
\begin{align}
    \cW_l^* \circ \cP|_{\Pi_{1,l}}&=\id_{B_{1,l}}\otimes \cR_{\delta_1}, \\
\cW_l^* \circ \cQ|_{\Pi_{1,l}}&=\cR_{\omega_l}.
\end{align}
Let $\ket{v_l}\in B_{1,l}$ be a unit vector with $q_l:= |\langle 0|v_l\rangle|^2$. A direct calculation gives
\begin{align}
\hat{X}_{1,1}^{(v_1)} := \sqrt{\delta_1} \Tr_{B_{1,1}}\!\bigl[(v_1\otimes \iden_{C_1})\,\omega_1^{-1/2}\bigr]\sqrt{\delta_1} &=
\begin{pmatrix}
\sqrt2\Bigl(q_1+\frac{1-q_1}{\sqrt3}\Bigr) & 0\\
0 & 1
\end{pmatrix}, \\
\hat{X}_{1,2}^{(v_2)}:=\sqrt{\delta_1}\Tr_{B_{1,2}}\!\bigl[(v_2\otimes \iden_{C_1})\,\omega_2^{-1/2}\bigr]\sqrt{\delta_1}
&=
\begin{pmatrix}
1 & 0\\
0 & \sqrt2\Bigl(q_2+\frac{1-q_2}{\sqrt3}\Bigr)
\end{pmatrix}.
\end{align}
Hence, for $\sD_2 (1,l)$ defined in Eq.~\eqref{eq:Dkl-def}, we get
\begin{equation}
2^{\widetilde D_2(1,1)}
= \sup_{v_1 \in \cD(B_{1,1})}
\Tr\!\bigl[(\hat{X}^{(v_1)}_{1,1})^2\bigr]
= \max_{q_1\in [0,1]} 
2\Bigl(q_1+\frac{1-q_1}{\sqrt3}\Bigr)^2+1 =3,
\end{equation}
and similarly
\begin{equation}
2^{\widetilde D_2(1,2)}
= \max_{q_2\in [0,1]}
2\Bigl(q_2+\frac{1-q_2}{\sqrt3}\Bigr)^2+1 = 3,
\end{equation}
with equality at $q_l=1$. Therefore, the right-hand side of Theorem~\ref{theorem:PvsQ} equals
\begin{equation}
\max_k \log\!\Bigl(\sum_l 2^{\widetilde D_2(k,l)}\Bigr)
=
\log(3+3)
=
\log 6.
\end{equation}

\smallskip
\noindent
\emph{The full channel divergence.} \\
Now, let
\begin{equation}
|\psi\rangle=\sqrt p\,|\psi_1\rangle\oplus \sqrt{1-p}\,|\psi_2\rangle\in B_{1,1} \oplus B_{1,2} = D_1,
\end{equation}
where $ 0\le p\le 1$ and $\ket{\psi_l}\in B_{1,l}$ are unit vectors. Let $q_l:= |\langle 0|\psi_l\rangle|^2$. Then, the corresponding operators $X^{(\psi)}_{1,1},X^{(\psi)}_{1,2}$ from Eq.~\eqref{eq:Xkl} are
\begin{align}
    X_{1,1}^{(\psi)} &=
\sqrt{p} \sqrt{\delta_1}
\Tr_{B_{1,1}}\!\bigl[(\psi_1\otimes \iden_{C_1})\,\omega_1^{-1/2}\bigr] \sqrt{\delta_1}
=
\begin{pmatrix}
\sqrt{2p} \left(q_1+\frac{1-q_1}{\sqrt3} \right) & 0\\
0 & \sqrt p
\end{pmatrix}, \\
X^{(\psi)}_{1,2} &=
\sqrt{1-p}\sqrt{\delta_1}
\Tr_{B_{1,2}}\!\bigl[(\psi_2\otimes \iden_{C_1})\,\omega_2^{-1/2}\bigr] \sqrt{\delta_1}  =
\begin{pmatrix}
\sqrt{1-p} & 0 \\
0 & \sqrt{2(1-p)}\left(q_2+\frac{1-q_2}{\sqrt3}\right)
\end{pmatrix}. \nonumber
\end{align}
Denoting $t(q):= q + (1-q)/\sqrt{3}$, we get (see Eq.~\eqref{eq:Xkl-sum}):
\begin{align*}
2^{\widetilde D_2(\cP\|\cQ)}
&= \sup_{\psi \in \cD (D_1)}
\Tr(X^{(\psi)}_{1,1}+X^{(\psi)}_{1,2})^2 \\
&= \max_{p\in [0,1]} \max_{q_l\in [0,1]}
\bigl(\sqrt{2p}\,t(q_1)+\sqrt{1-p}\bigr)^2
+
\bigl(\sqrt p+\sqrt{2(1-p)}\,t(q_2)\bigr)^2 \\
&= \max_{p\in [0,1]}
\bigl(\sqrt{2p}+\sqrt{1-p}\bigr)^2
+
\bigl(\sqrt p+\sqrt{2(1-p)}\bigr)^2 \\
&=
3+4 \max_{p\in [0,1]}\sqrt{2p(1-p)}
= 3+2\sqrt2.
\end{align*}
Both maxima are attained for $p=\frac12$ and $q_1=q_2=1$, i.e.\ for
\begin{equation}
|\psi\rangle=\frac{1}{\sqrt2}\Bigl(|0\rangle_{B_{1,1}}\oplus |0\rangle_{B_{1,2}}\Bigr).
\end{equation}
Hence, combining the two computations, we obtain
\begin{equation}
\sD_2(\cP\|\cQ)=\log(3+2\sqrt2)<\log 6
=
\max_k \log\!\Bigl(\sum_l 2^{\sD_2(k,l)}\Bigr).
\end{equation}

\subsection{Auxiliary lemmas} \label{appen:aux}

\begin{lemma}
For any set of real numbers $\{c_i\}_i$, the following holds true:
\begin{equation}
    \max_{\mu_i\geq 0, \sum_i \mu_i =1 } \left( H(\{\mu_i \}_i) + \sum_i \mu_i c_i \right) = \log \left( \sum_i 2^{c_i} \right).
\end{equation}
\end{lemma}
\begin{proof}
    Let $Z=\sum_i 2^{c_i}$ and $q_i = 2^{c_i}/Z$. Then, for any probability distribution $\{\mu_i\}_i$:
    \begin{align}
        H(\{\mu_i \}_i) + \sum_i \mu_i c_i &= \sum_i -\mu_i \log \mu_i + \sum_i \mu_i \log (q_iZ) \\
        &= \log Z - \sum_i \mu_i \log (\mu_i/q_i) \\
        &= \log Z - D(\mu \Vert q),
    \end{align}
    where $D(\mu \Vert q)\geq 0$ denotes the classical KL-divergence between $\mu$ and $q$. Hence,
    \begin{align}
        H(\{\mu_i \}_i) + \sum_i \mu_i c_i \leq \log Z,
    \end{align}
    and equality holds when $\mu=q$.
\end{proof}

\begin{lemma}
    For any set of strictly positive numbers $\{c_i \}_i$, the following holds true:
\begin{equation}
    \min_{\mu_i\geq0, \sum_i \mu_i=1}\sum_i \mu_i^2 c_i = \frac{1}{\sum_i 1/c_i}.
\end{equation}
\end{lemma}
\begin{proof}
    Using Cauchy-Schwarz inequality, we get for any probability distribution $\{\mu_i\}_i$:
\begin{equation}
    1 = \left(\sum_i \mu_i \right)^2 = \left( \sum_i \frac{\mu_i \sqrt{c_i}}{\sqrt{c_i}}  \right)^2 \leq \sum_i \mu_i^2 c_i \sum_i\frac{1}{c_i},
\end{equation}
and the minimum is achieved by choosing the probability vector with components
\begin{equation}
    \mu_i = \frac{1/c_i}{\sum_j 1/c_j} \qquad \text{for every }i.
\end{equation}
\end{proof}

\begin{lemma}
    For any set of non-negative numbers $\{c_l \}_l$ with $\sum_l c_l>0$ and $\alpha>1$, the following holds true:
\begin{equation}
    \max_{p_l\geq 0, \sum_l p_l=1} \frac{\alpha}{\alpha-1} \log \left( \sum_l p_l^{\frac{1}{\alpha}} c_l^{\frac{\alpha-1}{\alpha}} \right) = \log \left( \sum_l c_l \right).
\end{equation}
\end{lemma}
\begin{proof}
    Using H\"older's inequality, we get for any probability distribution $\{p_l\}_l$:
    \begin{align}
        \sum_l p_l^{\frac{1}{\alpha}} c_l^{\frac{\alpha-1}{\alpha}} &\leq \left( \sum_l p_l \right)^{\frac{1}{\alpha}} \left(\sum_l c_l \right)^{\frac{\alpha-1}{\alpha}} \\
        &= \left(\sum_l c_l \right)^{\frac{\alpha-1}{\alpha}}.
    \end{align}
    Hence, 
    \begin{equation}
        \frac{\alpha}{\alpha-1} \log \left( \sum_l p_l^{\frac{1}{\alpha}} c_l^{\frac{\alpha-1}{\alpha}} \right) \leq \log \left( \sum_l c_l \right).
    \end{equation}
    Moreover, equality is achieved by choosing the probability vector with components
    \begin{equation}
        p_l = \frac{c_l}{\sum_j c_j} \qquad \text{for every }l.
    \end{equation}
\end{proof}

\end{subappendices}

\bibliographystyle{plainurl}
\bibliography{references, references_zotero}

\end{document}